\documentclass{article}
\usepackage{epsf}
\usepackage{latexsym}
\usepackage{alltt}

\newbox\isgreater
\setbox\isgreater\hbox{\raise 3.4pt\hbox{${?\atop{\textstyle \ge}}$}}
\ht\isgreater=0pt
\newtheorem{theorem}{Theorem}

\newtheorem{corollary}{Corollary}

\newtheorem{proposition}{Proposition}
\newenvironment{proof}{\noindent{\sc Proof.\space}\ignorespaces}
{\hfill$\Box$\par\smallskip}
\newenvironment{remark}{\noindent{\sc Remark:\space}\ignorespaces}%
{\par\smallskip}
\newcommand{\E}{\mbox{\rm E}}

\newsavebox{\savepar}
\newenvironment{boxit}{\begin{lrbox}{\savepar}
  \begin{minipage}[b]{4.3in}}
  {\end{minipage}\end{lrbox}\fbox{\usebox{\savepar}}}

\begin{document}

\title{Cache Analysis of Non-uniform Distribution Sorting Algorithms}
\author{{\small Naila Rahman}\\{\small Department of Computer Science}\\{\small University of Leicester}\\
{\small Leicester LE1 7RH, UK.}\\ \texttt{{\small naila@mcs.le.ac.uk}}\\ 
\\ 
{\small Rajeev Raman}\thanks{Supported in part by EPSRC grant GR/L92150}\\
{\small Department of Computer Science}, \\ {\small University of Leicester},\\
{\small Leicester LE1 7RH, UK.}\\ \texttt{{\small r.raman@mcs.le.ac.uk}}.}




\maketitle

\begin{abstract}
We analyse the average-case cache performance of distribution sorting
algorithms in the case when keys are independently but not necessarily 
uniformly distributed.  The analysis is for both `in-place' and 
`out-of-place' distribution sorting algorithms and is more accurate than 
the analysis presented in~\cite{RRESA00}. In particular, this new analysis 
yields tighter upper and lower bounds when the keys are drawn from a 
uniform distribution.
We use this analysis to tune the performance of 
the integer sorting algorithm MSB radix sort when it is used to 
sort independent uniform floating-point numbers (floats).  
Our tuned MSB radix sort algorithm comfortably outperforms 
a cache-tuned implementations of bucketsort \cite{RR99} and Quicksort
when sorting uniform floats from $[0, 1)$.

\end{abstract}

\section{Introduction}
Distribution sorting is a popular alternative to comparison-based
sorting which involves placing $n$ input keys into $k \le n$ 
{\it classes\/} based on their value \cite{Knuth}. 
The classes are chosen so that all the keys in the $i$th
class are smaller than all the keys in the $(i+1)$st class,
for $i = 1, \ldots, k-1$, and furthermore, the class to which
a key belongs can be computed in $O(1)$ time (e.g. if the
keys are floats in the range $[a, b)$, we can calculate 
the class of a key $x$ as  $1 + \lfloor \frac{x-a}{b-a} \cdot k \rfloor$).
Thus, the original sorting problem is reduced in linear time 
to the problem of sorting the keys in each class.  A number of
distribution sorting algorithms have been developed which run in
linear (expected) time under some assumptions about the input keys, 
such as bucket sort and radix sort.  
Due to their poor {\it cache\/} utilisation, even good
implementations---which minimise instruction counts---of 
these `linear-time' algorithms fail to outperform general-purpose
$O(n \log n)$-time algorithms such as Quicksort or Mergesort
on modern computers \cite{LL97,RR99}. 

Most algorithms 
are
based upon the random-access machine model \cite{AHU}, 
which assumes that main memory is as fast as the CPU.
However, in modern computers, main memory is typically 
one or two orders of magnitude slower than the CPU \cite{Hennessy}.
To mitigate this, one or more
levels of {\it cache\/} are introduced between CPU and memory.  A cache is a
fast associative memory which holds the values of some
main memory locations. If the CPU requests the contents of a
memory location, and the value of that location is held in 
some level of cache (a cache {\it hit}), the CPU's request is 
answered by the cache itself in typically 1-3 clock cycles; 
otherwise (a cache {\it miss}) it is answered by accessing 
main memory in typically 30-100 clock cycles.  
Since typical programs exhibit {\it locality of reference\/} \cite{Hennessy}, 
caches are often effective.  However, algorithms such
as distribution sort have poor locality of reference,
and their performance can be greatly improved 
by optimising their cache behaviour. 
A number of papers have recently addressed this issue 
\cite{LFL99,LL97,RR99,RR00,Mehlhorn2000,SC99},
mostly in the context of sorting and related problems.
There is also a large literature on algorithms
specifically designed for hierarchical models of memory \cite{VitterSurvey,ArgeSurvey},
but there are some important differences between these models and ours
(see \cite{RahmanThesis} for a summary).

The cache performance of comparison-based 
sorting algorithms was studied in \cite{LL97,Mehlhorn2000,SC99} 
and  distribution sorting algorithms were considered 
in \cite{LFL99,LL97,RR99}.  
One pass of a distribution sort consists of a {\it count\/} phase
where the number of keys in each class are determined, followed by
a {\it permute\/} phase where the keys belonging to the same class
are moved to consecutive locations in an array.
We give an analysis of the cache behaviour of
the permute phase, assuming the keys are independently drawn
from a non-uniform distribution. 
In~\cite{RRESA00} we focused on `in-place' permute, where the keys 
are rearranged without  placing them first in an auxiliary 
array. In this paper we extend the analysis to `out-of-place' permute.
We model the above algorithms as probabilistic
processes, and analyse the cache behaviour of these processes.
For each process we give an exact expression for, as well as matching
closed-form upper and lower bounds on, the number of misses. 

In previous work on the cache analysis of distribution sorting,
\cite{LFL99} have analysed the (somewhat easier) count phase for 
non-uniform  keys, and \cite{RR99} gave an empirical analysis
of the permute phase for uniform keys.  The process of 
{\it accessing multiple sequences\/} of memory locations,
which arises in multi-way merge sort, was
analysed previously by \cite{Mehlhorn2000,SC99}.
The analysis in \cite{Mehlhorn2000} assumes that
accesses to the sequences are controlled by an adversary;
our analysis demonstrates, among other things, that with 
uniform randomised accesses to
the sequences, more sequences can be accessed optimally.
In~\cite{SC99} a lower bound on cache misses is given for 
uniform randomised accesses; our lower bound is somewhat sharper.
The analysis also improves upon the results in~\cite{RRESA00},
by giving tighter upper and lower bounds when the keys are 
drawn from a uniform distribution.

In practice there are often cases when keys are not uniform
(e.g., they may be normally distributed); our analysis can
be used to tune distribution sort in these cases.  
We consider a different application here:
sorting uniform floats using an {\it integer\/} sorting algorithm.  
It is well known that one can sort floats by sorting the bit-strings 
representing the floats, interpreting them as integers \cite{Hennessy}.   
Since (simple) operations on integers are faster
than operations on floats,  this can improve performance; indeed, 
in \cite{RR99} it was observed that an {\it ad hoc\/} implementation 
of the integer sorting algorithm {\it most-significant-bit first\/} 
radix sort (MSB radix sort) outperformed an optimised version of 
bucket sort on uniform floats.   
We observe that a uniform distribution on
floating-point numbers induces a non-uniform distribution on 
the representing integers, and use our cache analysis to improve
the performance of MSB radix sort on our machine. 
Our tuned `in-place' MSB radix sort
comfortably outperforms optimised implementations of other 
in-place or `in-place' algorithms such as Quicksort or 
MPFlashsort~\cite{RR99}, which is a cache-tuned version
of bucket sort.

\section{Cache preliminaries}

This section introduces some terminology and
notation regarding caches.
The size of the cache is normally expressed in terms of
two parameters, the {\it block size} ($B$) and the number of
{\it cache blocks\/} ($C$).   We consider main memory as being
divided into equal-sized {\it blocks} consisting of $B$
consecutively-numbered memory locations, with blocks starting at locations
which are multiples of $B$.  The cache is also divided
into blocks of size $B$; one cache block can hold the value of
exactly one memory block. Data is moved to and from
main memory only as blocks.

In a {\it direct-mapped\/} cache, the value of memory location
$x$ can only be stored in cache block
$c = {{(x \mbox{\rm\ div\ } B)}\bmod{C}}$.
If the CPU accesses location $x$ and cache block
$c$ holds the values from $x$'s block
the access is a cache hit; otherwise it is
a cache miss and the contents of the
block containing $x$ are copied into cache block $c$, {\it evicting}
the current contents of cache block $c$.
For our purposes, cache misses can be classified
into {\it compulsory misses}, which occur when a
memory block is accessed for the first time, 
{\it capacity misses}, which occurs on an access to a 
memory block that was previously evicted because the 
cache could not hold all the blocks being actively 
accessed, 
and {\it conflict misses}, which happen when a block
is evicted from cache because another memory block
that mapped to the same cache block was accessed.

\section{Distribution sorting}
\label{sec:GenDistSort}

As noted in the introduction, a distribution pass has two main phases,
a {\it count} phase and a {\it permute} phase, and our focus here
is on the latter. 

While describing this algorithm, the term {\it data\/} array refers to
the array holding the input keys, and the term {\it count\/} 
refers to an auxiliary array used by these algorithms.
Each pass consists of two main phases, a {\it count} phase followed
by a {\it permute} phase. 

The count phase counts for class $1 \le i \le k-1$, the total number of keys in
classes $0, \ldots, i-1$. For class $i=0$ this cumulative count is 0.
Ladner~et~al~\cite{LFL99} give an analysis of the count phase of 
distribution sorting on a direct-mapped cache for uniformly and randomly 
distributed keys.

There are two main variants of the permute phase, in the first variant keys 
are permuted from the data array to the auxiliary {\it destination\/} array, 
this is called an {\it out-of-place permutation}. 
In the second variant keys in the data array are permuted within the
data array, this is called an {\it in-place permutation}.

\subsection{Permute phase}
The permute phase uses the cumulative count of keys generated during
the count phase, to permute the keys to their respective classes.
We now describe the two variants of the permute phase.
In the description below it is assumed that $k$ has been
appropriately initialised, and that the function {\tt classify}
maps a key to a class numbered $\{0,\ldots, k-1\}$ in $O(1)$ time.


\subsubsection{Out-of-place permute}
During an out-of-place permute, for any class $j$,
unless all elements of that class have already been moved,
{\tt COUNT[$j$]} points to the leftmost (lowest-numbered)
available location for an element of class $j$ in 
an $n$ element auxiliary array, {\tt DEST}.
\begin{figure}
\hspace*{15ex}{Permute phase(out-of-place permutation)}\\
{
\tt
\hspace*{15ex}1~for $i$ := $0$ to $n - 1$ do\\
\hspace*{21ex}{\it key} := DATA[$i$];\\
\hspace*{21.15ex}$x$ := classify({\it key})\\
\hspace*{21.15ex}{\it idx} := COUNT[$x$];\\
\hspace*{21.15ex}COUNT[$x$]++;\\
\hspace*{21.15ex}DEST[idx] :=  {\it key};\\
}
\caption{Permute phase for an `out-of-place' permutation in
a generic distribution sorting algorithm.
{\tt DATA} holds the input keys. {\tt COUNT} and {\tt DEST}
are auxiliary arrays. }
\label{fig:OutofplacePermuteCode}
\end{figure}
Figure~\ref{fig:OutofplacePermuteCode} shows the pseudo-code for out-of-place 
permutation.
In Step 1, for each element in {\tt DATA}:
we determine its class; 
using the count array we determine the next available location for
this key in the {\tt DEST} array;
we increment the count array, thus setting the location for
the next key of the same class;
finally we move the key to its location in {\tt DEST}.
Since each step takes constant time, this 
out-of-place permutation takes $O(n)$ time whenever $k\le n$.

\subsubsection{In-place Permute}
The in-place permutation strategy described here is similar to that 
described by Knuth~\cite[Soln 5.2-13]{Knuth}.
Before an in-place permute phase begins, a copy of the count array is made in
a $k$ element auxiliary {\it start} array.  During the permute phase, for any 
class $j$,
an invariant is that locations $\mbox{\tt START[$j$]}, 
\mbox{\tt START[$j$]} + 1, \ldots, \mbox{\tt COUNT[$j$]}- 1$
contain elements of class $j$, i.e.
{\tt COUNT[$j$]} points to the leftmost (lowest-numbered)
 available location for
an element of class $j$.  Thus, for $j = 0, \ldots, k-2$,
all elements of class $j$ have been permuted 
if $\mbox{\tt COUNT[$j$]} \ge \mbox{\tt START[$j+1$]}$,
and such a class will be called {\it complete\/} in what follows.
Class $k-1$ is complete when $\mbox{\tt COUNT[$k-1$]} \ge n$.
\begin{figure}
\hspace*{15ex}{Permute phase(in-place permutation)}\\
{
\tt
\hspace*{15ex}1~{\it leader} := $n-1$; \\
\hspace*{15ex}2~{\it idx} := {\it leader}; {\it key} := DATA[{\it idx}];\\
\hspace*{15ex}3.1~{\it x} := classify({\it key}); \\
\hspace*{15ex}3.2~{\it idx} := COUNT[{\it x}];\\
\hspace*{15ex}3.3~COUNT[{\it x}]++; \\
\hspace*{15ex}3.4~swap {\it key} and DATA[{\it idx}];\\
\hspace*{15ex}3.5~if {\it idx} $\not =$ {\it leader} repeat 3.1;\\
\hspace*{15ex}4~while ($x > 0 \wedge \mbox{\tt COUNT[$x-1$]} \ge$ START[$x$]) \\
\hspace*{21ex}$x$--;\\
\hspace*{15ex}5~if ($x > 0$) {\it leader} := START[$x$]$-1$; \\
\hspace*{21ex}go to 2;
}
\caption{Permute phase for an `in-place' permutation in
a generic distribution sorting algorithm.
{\tt DATA} holds the input keys. {\tt COUNT}  and {\tt START}
are auxiliary arrays. After the count phase, {\tt COUNT} 
is copied into {\tt START}.}
\label{fig:InplacePermuteCode}
\end{figure}
Figure~\ref{fig:InplacePermuteCode} shows the pseudo-code for in-place 
permutation.
We now describe this permutation, which consists of two
main activities: {\it cycle following\/} and 
{\it cycle leader finding}.  
In cycle following, keys are moved to their final
destinations in the data array along a cycle in the permutation 
(Steps 2 and 3).
Once a cycle is completed, we move to cycle leader finding, where
we find the `leader' (index of the rightmost element) of the next cycle
(Steps 1, 4 and 5).  
A cycle leader is simply the rightmost location of
the highest-numbered incomplete class.  By the definition of
a complete class, initially the leader must be position $n-1$.
In more detail, the steps are as follows:

\begin{itemize}
\item In Step 1 $n-1$ is selected as the first cycle leader.

\item In Step 2 the key at the leader's position is copied into the
variable $key$, thus leaving a `hole' in the leader's position.

\item In Steps 3.1-3.5 the key $key$ is swapped with the
key at $key$'s final position.  If $key$ `fills the hole',
the cycle is complete, otherwise we repeat these steps.

\item In Step 4 the algorithm searches for a new cycle leader.  Suppose
the leader of the cycle which just completed was 
the last location of class $j$.  When this cycle ends, class $j$ must also
be complete, as a key of class $j$ has been moved into the last
location of class $j$.  Note that classes $j+1,j+2,\ldots$
must already have been complete when the leader of this cycle
was found.   Note that
the program variable $x$ has value $j$ at the end of this cycle,
so the search for the next leader
begins with class $j-1$, counting down (Step 4).

\item In Step 5 we check to see if all classes have completed and terminate
if this is the case.

\end{itemize}

Clearly the in-place permutation in one pass of distribution sorting takes 
$O(n)$ time whenever $k \le n$.  


\section{Cache analysis}
\label{sec:PermutePhase}
We now analyse cache misses in a direct-mapped cache
during the permute phase of distribution sorting when the keys
are independently drawn from a non-uniform random distribution.
In the permute phase of distribution sorting, when a key is moved to its 
destination, the algorithms described in Section~\ref{sec:GenDistSort}
access any one of $k$ elements in the {\tt COUNT} array and any one of
$k$ locations in the {\tt DATA} or {\tt DEST} arrays, depending on
whether the permutation is in-place or out-of-place. The actual
locations accessed are dependent on the value of the permuted key,
so, if the keys are independently and randomly distributed then, for every 
key permuted there are two random accesses to memory, one in the count
array and one in {\tt DATA} or {\tt DEST}. These random accesses can
potentially lead to a large number of cache conflict misses.

Our approach is to define two continuous
processes which model in-place and out-of-place permutations.
Process ``in-place" models an
in-place permutation and is shown in Figure~\ref{fig:ProcInplace},
and Process ``out-of-place" models an
out-of-place permutation and is shown in Figure~\ref{fig:ProcOutplace}.
Each round of a process models the permutation of a key to its
destination, and we analyse the expected number of cache misses in 
$n$ rounds of these processes.
Our precise equations are difficult to compute so we 
also give closed-form upper and lower bounds on these precise equations.
We use our results for in-place permutations 
to get upper and lower bounds on the expected number
of cache misses in a process which models accesses to multiple sequences.

The assumptions in the processes mean that we have to access at least $n$ 
distinct locations in memory, which requires $\Omega(n/B)$ cache misses. 
In the analysis, we will say that a process is optimal if it incurs $O(n/B)$ 
cache misses.
In distribution sorting, the larger the value of $k$, the fewer the number
of passes over the data, hence the fewer the capacity misses. As we will see,
if $k$ is too large, then there can be a large number of conflict misses.
The aim of the analysis is to determine
the largest value of $k$, for a particular distribution of keys,
such that there are $O(n/B)$ misses in one pass of distribution sorting.

\subsection{Processes}
We now give the two processes which model the distributing
of keys drawn independently and randomly from a non-uniform
distribution into $k$ classes.

\subsubsection{Process to model an in-place permutation}
\label{proc:Inplace}
Let $k$ be an integer, $2\le k \le CB$.
We are given $k$ probabilities $p_1,\ldots,p_k$, such that
$\sum_{i=1}^k p_i = 1$.  The process maintains $k$ pointers 
$D_1, \ldots, D_k$, and there are also 
$k$ consecutive `count array' locations, ${\cal C} = 
c_1, \ldots, c_k$.  The process (henceforth called
{\it Process ``in-place"}) executes a sequence of
{\it rounds}, where each round consists in performing steps 1-3 below:\\

\medskip
\noindent
\begin{center}
\begin{boxit}
\begin{minipage}{4.2in}
\noindent
\center{\textsf{Process ``in-place"}}

\begin{enumerate}
\item Pick an integer $x$ from $\{1,\ldots,k\}$ such 
that $\Pr[x = i] = p_i$, independently of all previous picks.

\item Access the location $c_x$.

\item Access the location pointed to by $D_x$, increment $D_x$ by 1.
\end{enumerate}
\end{minipage}
\end{boxit}
\end{center}

\medskip

We denote the locations accessed by the pointer $D_i$ by
$d_{i,1}, d_{i,2}, \ldots$, for $i = 1, \ldots, k$.  
We assume that:
\begin{itemize}

\item[(a)] the start position of each pointer is uniformly and independently 
distributed over the cache, i.e., for each $i$, 
${{d_{i,1}}\bmod{BC}}$ is uniformly and independently 
distributed over $\{0,\ldots, BC-1\}$,

\item[(b)] during the process, the pointers traverse sequences
of memory locations which are disjoint from each other and from $\cal C$,

\item[(c)] $c_1$ is located on an aligned block boundary, i.e., 
${{c_1}\bmod{B}} = 0$, 

\item[(d)] the pointers $D_i$, for $i = 1, \ldots, k$, are in separate
memory blocks.

\end{itemize}

\begin{figure}
\epsfxsize 4in
\centerline{\epsfbox{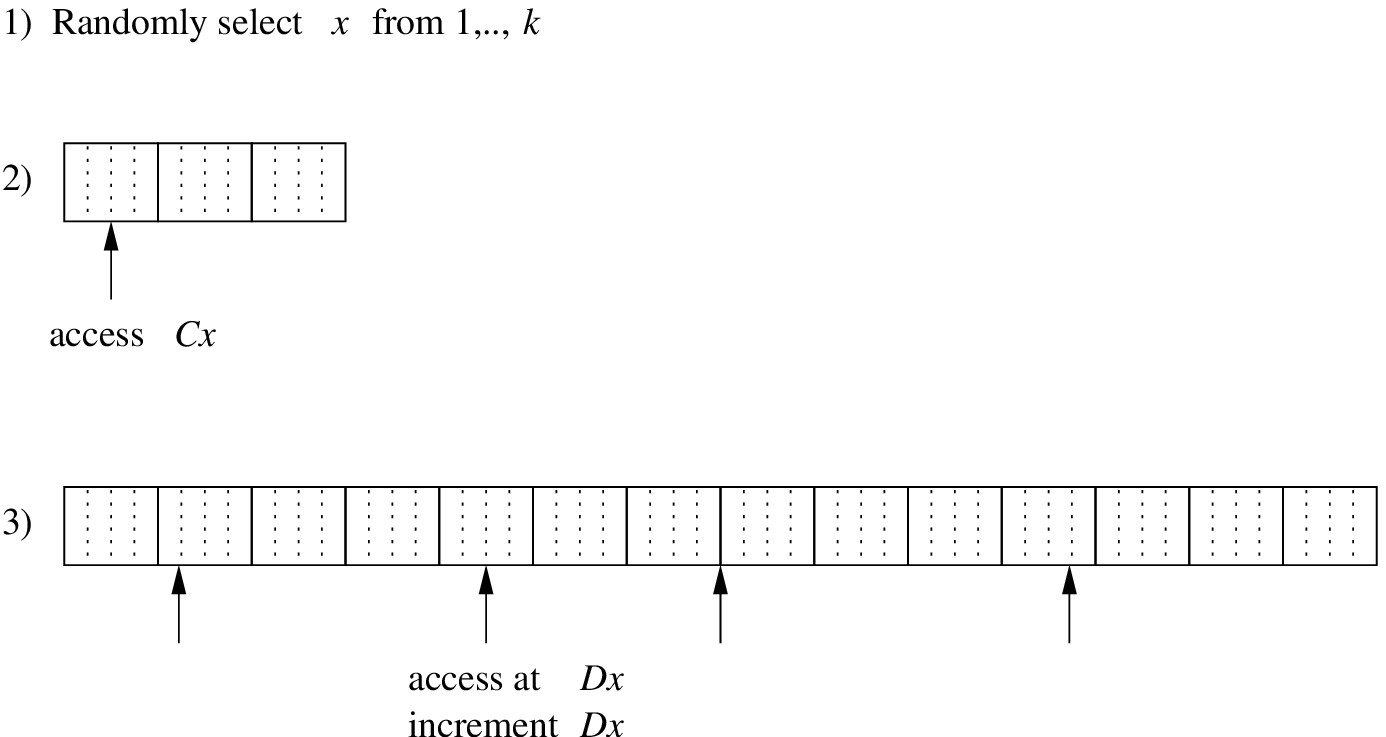}}
\caption{Process ``Inplace".}
\label{fig:ProcInplace}
\end{figure}

Assuming that the cache is initially empty, the objective is to
determine the expected number of cache misses incurred by the
above process over $n$ rounds, with the expectation taken 
over the random choices in Step 1 as well as the starting
positions of the pointers.  

\subsubsection{Process to model an out-of-place permutation}
\label{proc:Outplace}
This process is like Process ``in-place",
but it is augmented with accesses to a sequence of consecutive locations
in a source array, ${\cal S}$, determined by an index $s$.
The process, henceforth called {\it Process ``out-of-place"}, executes a sequence of
{\it rounds}, where each round consists in performing steps 1-4 below:

\medskip
\noindent
\begin{center}
\begin{boxit}
\begin{minipage}{4.2in}
\noindent
\center{\textsf{Process ``out-of-place"}}
\begin{enumerate}
\item Access the location ${\cal S}[s]$, increment $s$ by 1.

\item Pick an integer $x$ from $\{1,\ldots,k\}$ such 
that $\Pr[x = i] = p_i$, independently of all previous picks.

\item Access the location $c_x$.

\item Access the location pointed to by $D_x$, increment $D_x$ by 1.

\end{enumerate}
\end{minipage}
\end{boxit}
\end{center}

\medskip

We make assumptions (a), (c), and (d)  from Process ``in-place", 
assumption (b) is modified as below and we add a further assumption:
\begin{itemize}
\item[(b)] during the process, the pointers traverse sequences
of memory locations which are disjoint from each other,
from ${\cal C}$ and from ${\cal S}$.
\end{itemize}

\begin{figure}
\epsfxsize 4in
\centerline{\epsfbox{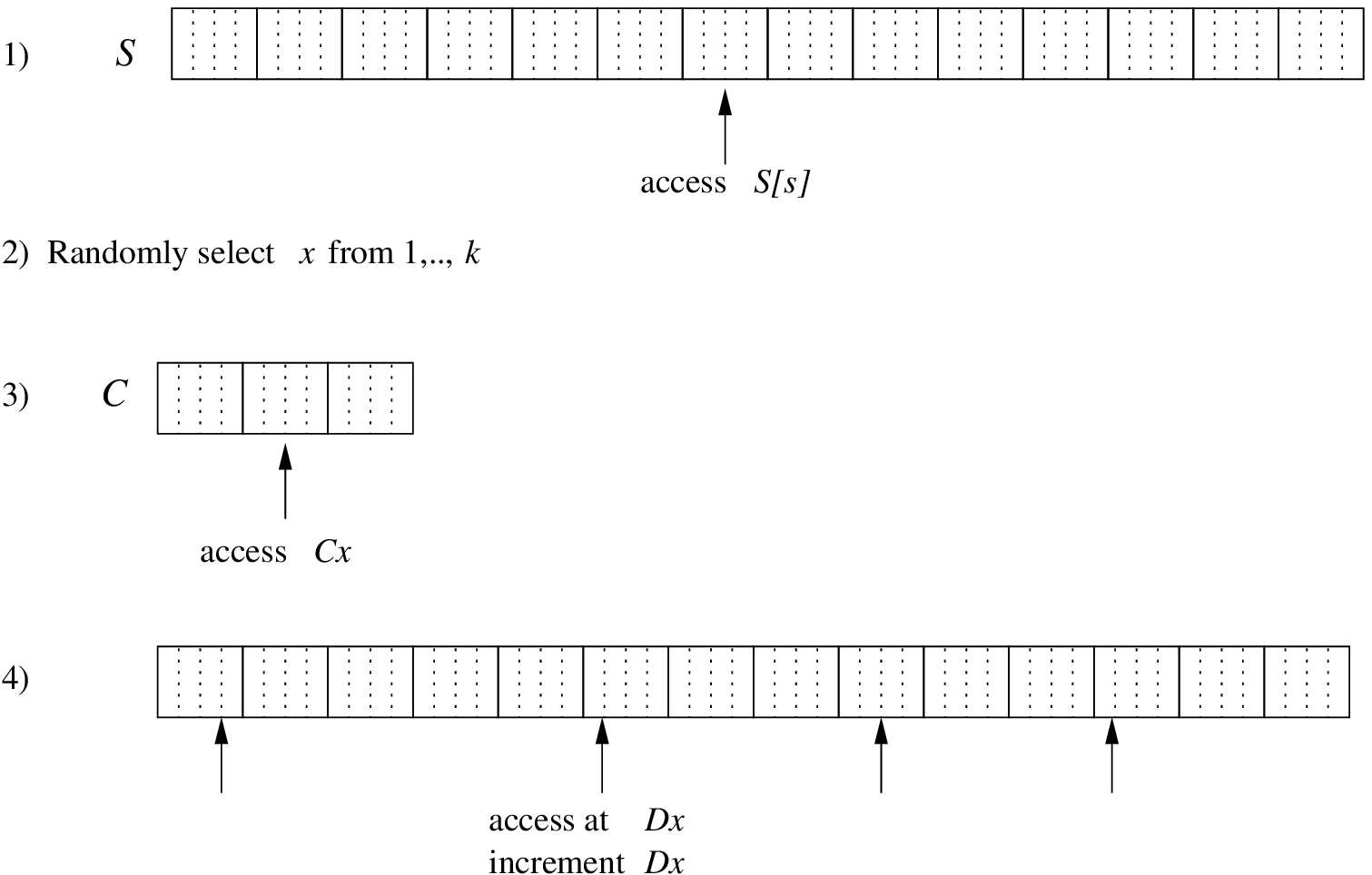}}
\caption{Process ``Out-of-place".}
\label{fig:ProcOutplace}
\end{figure}

Assuming that the cache is initially empty, again the objective is to
determine the expected number of cache misses incurred by the
above process over $n$ rounds, with the expectation taken 
over the random choices in Step 2 as well as the starting
positions of the pointers.

\subsection{Preliminaries}

We now introduce some notation that will be used for the
analysis.
We use $k$ to denote the number of classes that the keys will be
distributed into, and throughout the analysis we assume that $B$ divides $k$.
Assume that we are given a set of $k$ probabilities 
$p_1, \ldots, p_k$, such that $\sum_{i=0}^{k} p_i = 1$.
The expected value of a function $f$ of a
random variable $X$ is denoted as $\E[f(X)]$.  When
we wish to make explicit the distribution $D$ from which the random
variable is drawn, we will use the notation $\E_{X \sim D}[f(X)]$.
All vectors have dimension $k$ (the
number of classes) unless stated otherwise, and
we denote the components of a vector $\bar{x}$ by $x_1$,
$x_2, \ldots, x_k$.  We now define some probabilities:

\smallskip
\noindent
(i)~For all $i \in \{1,\ldots,k/B\}$, $P_i=\sum_{l=(i-1)B+1}^{iB} p_l$.

\smallskip
\noindent
(ii)~For all $i \in \{1,\ldots,k\}$, we denote by $\bar{a^{i}}$ the
following vector: 
$a^{i}_j = 0$ if $i = j$, and $a^{i}_j = p_j/(1 - p_i)$ otherwise 
and
by $\bar{b^{i}}$ the
following vector: 
$b^{i}_j = 0$ if $(i-1)B+1 \le j \le iB$, and 
$b^{i}_j = p_j/(1 - P_i)$ otherwise. 
(Note that $\sum_j a^{i}_j = \sum_j b^{i}_j = 1$).

\smallskip
Let $m\ge 0$
be an integer and $\bar{q}$ be a vector of non-negative reals such that
$\sum_{i} q_i = 1$.
We denote by $\varphi(m,\bar{q})$ the 
probability distribution on the number of balls
in each of $k$ bins, when $m$ balls are independently put into these
bins, and a ball goes in bin $i$ with probability $q_i$, 
for $i \in \{1, \ldots, k\}$.  Thus, $\varphi(m,\bar{q})$ is a distribution
on vectors of non-negative integers. 
If $\bar{\mu}$ is drawn from $\varphi(m,\bar{q})$, then:
\begin{equation}
\Pr[\mu_1 = m_1, \ldots, \mu_k = m_k] = 
\left ( \prod_{j=1}^{k} q_{j}^{m_j}\right ) m! / \prod_{j=1}^{k} m_j!
\label{eq:PrPartition}
\end{equation}
whenever $\sum_{i=1}^k m_i = m$; all other vectors have zero 
probability\footnote{We take $0^0 = 1$ in 
\protect{Eq.~\ref{eq:PrPartition}}.}.
%
We now define functions $f(x)$ for $x \ge 0$ and $g(\bar{m})$
for a vector $\bar{m}$ of non-negative integers:
\begin{eqnarray}
f(x) &  = & \left \{ \begin{array}{ll}
1 & \mbox{\rm \ if $x = 0$}, \\
1 - \frac{x + B - 1}{BC} & \mbox{\rm \ if $0 < x \le BC - B + 1$}, \\
0 & {\rm \ otherwise}.
\end{array} \right .
\label{def:fx}\\
g(\bar{m}) &  = & \frac{1}{C} \sum_{i=1}^{k/B} \min\{1, \sum_{l=(i-1)B+1}^{iB} m_l\}.
\label{def:gx}
\end{eqnarray}

%

We now set out some propositions that are used in the proofs.\\

\begin{proposition}
\label{prop:ProdSum}
For all real numbers $x_i, i=1,\ldots,k$, such that $|x_i| \le 1$ we have that:
$$\prod_{i=0}^{k} (1-x_i) \ge 1- \sum_{i=0}^{k} x_i.$$
\end{proposition}

\begin{proposition}
\label{prop:GeoSeriesa}(a)
For all real numbers $x$, such that $|x| < 1$ we have that:
$$\sum_{m=0}^{\infty} x^m=\frac{1}{1-x}.$$

\label{prop:GeoSeriesb}(b)
For all real numbers $x$, such that $|x| < 1$  we have that:
$$\sum_{m=0}^{\infty} mx^m = \frac{x}{(1-x)^2}.$$

\label{prop:GeoSeriesc}(c)
For all real numbers $x$, such that $0 < x < 2$, we have that:
$$\sum_{m=0}^{\infty} x(1-x)^m m = \frac{1}{x} - 1.$$
\end{proposition}
\begin{proof}
Proposition~\ref{prop:GeoSeriesa}(a) is the standard summation for an infinite 
decreasing geometric series.
We obtain Proposition~\ref{prop:GeoSeriesa}(b)
by differentiating both sides of the equation in 
Proposition~\ref{prop:GeoSeriesa}(a).
Proposition~\ref{prop:GeoSeriesa}(c) is obtained using 
Proposition~\ref{prop:GeoSeriesb}(b)
and is the expected value of the geometric distribution multiplied by $1-x$.
\end{proof}

\begin{proposition}
\label{prop:prop3}
For all real numbers $p$ and $q$ such that $0<p-q<2$, we have that:
$$\sum_{m=0}^{\infty} p(1-p)^m \left (1-\frac{q}{1-p} \right )^m
=\frac{p}{p+q} . $$
\end{proposition}
\begin{proof}
Since $(1-p) \left (1-\frac{q}{1-p} \right) = 1-p-q$, using
Proposition~\ref{prop:GeoSeriesa}(a) we get that: 
$$\sum_{m=0}^{\infty} p(1-p)^m \left (1-\frac{q}{1-p} \right )^m
=\sum_{m=0}^{\infty} p \left (1-p-q \right )^m
=\frac{p}{p+q} . $$
\end{proof}

\begin{proposition}
\label{prop:Taylora}(a)
For all real numbers $x$, we have that:
$$e^{-x} \ge 1-x . $$

\label{prop:Taylorb}(b)
For all real numbers $x \ge 0$, we have that:
$$e^{-x} \le 1-x+\frac{x^2}{2} . $$

\label{prop:Taylorc}(c)
For all real numbers $x_i, i=1,\ldots,k$, such that $x_i \le 1 $ we have that:
$$ \prod(1-x_i) \le 1-\sum x_i + \sum \frac{x_i ^2}{2} . $$
\end{proposition}

\begin{proof}
Propositions~\ref{prop:Taylora}(a)~and~\ref{prop:Taylorb}(b) are from Taylor's series.
For Propositions~\ref{prop:Taylorc}(c) we use Proposition~\ref{prop:ProdSum}.
\end{proof}

\begin{proposition}
\label{prop:prop6}
For all real numbers $x$ and $y$, such that $x \le 1$ and $y \ge 0$, 
we have that:
$$e^{-xy} \ge (1-x)^y.$$
\end{proposition}
\begin{proof}
This proposition is proved using Proposition~\ref{prop:Taylora}(a).
\end{proof}

\begin{proposition}
\label{prop:SumProb0toya}(a)
For all real numbers $x$ and $p$ and integer $y$, such that $0 < p \le 1$,
$y \ge 0$ and $x ( 1/p+y ) = O(1)$, 
we have that:
\begin{eqnarray*}
\sum_{m=0}^{y} p(1-p)^m m x &=& x \left( \frac{1}{p}-1 \right) - O(e^{-py}).
\end{eqnarray*}

\label{prop:SumProb0toyb}(b)
For all real numbers $x$ and $p$ and integer $y$, such that $0 < p \le 1$,
$y \ge 0$ and $x = O(1)$, we have that:
\begin{eqnarray*}
\sum_{m=0}^{y} p(1-p)^m x &=& x - O(e^{-py}).
\end{eqnarray*}

\label{prop:SumProb0toyc}(c)
For all real numbers $x$, $p$ and $q$ and integer $y$, such that $0 < p-q < 2$,
$y \ge 0$ and $\frac{xp}{p+q} = O(1)$, we have that:
\begin{eqnarray*}
\sum_{m=0}^{y} p(1-p)^m x &=& \frac{xp}{p+q} - O(e^{-(p+q)y}).
\end{eqnarray*}

\label{prop:SumProb0toyd}(d)
For all real numbers $m$, $x$, $p$ and $q$ and integer $y$, such that $0 < p-q < 2$,
$y \ge 0$ and $\frac{xp}{p+q} (\frac{1-p-q}{p+q} +y+1) = O(1)$, we have that:
\begin{eqnarray*}
\sum_{m=0}^{y} p(1-p)^m x &=& \frac{xp(1-p-q)}{(p+q)^2} - O(e^{-(p+q)y}).
\end{eqnarray*}

\end{proposition}
Note that we are misusing the $O$ notation here to hide constant factors that
are independent of the variables in the equations.

\begin{proof}
Using Proposition~\ref{prop:GeoSeriesc}(c) and Proposition~\ref{prop:prop6},
Proposition~\ref{prop:SumProb0toyb}(a) is proved as follows:
\begin{eqnarray*}
\sum_{m=0}^{y} p(1-p)^m m x &=&
\sum_{m=0}^{\infty} p(1-p)^m m x  -
(1-p)^{y+1} \sum_{m=0}^{\infty} p(1-p)^m (m + y +1) x  \\
&=& x \left( \frac{1}{p}-1 \right) - (1-p)^{y+1} x \left(\frac{1}{p}+y \right)\\
&=& x \left( \frac{1}{p}-1 \right) - O(e^{-py}).
\end{eqnarray*}
The proofs of Propositions~\ref{prop:SumProb0toyb}(b),~\ref{prop:SumProb0toyc}(c) 
and~\ref{prop:SumProb0toyd}(d) are now trivial.
\end{proof}

\bigskip

The vector of random variables $X = (X_1, \ldots X_n)$, is
{\it negatively associated}~\cite{JP83} if for every two disjoint
index sets, $I, J \subset [n]$,
$$
\E[ f(X_i, i \in I) g(X_j, j \in J) ]
   \le \E[ f(X_i, i \in I)] \E[ g(X_j, j \in J)]
$$
for all functions $f: {\Re}^{|I|} \rightarrow \Re$ and 
$f: {\Re}^{|J|} \rightarrow \Re$
that are both non-decreasing or non-increasing.

\begin{proposition}
\label{prop:NegAssoc}
If the random variables $X_1, \ldots X_k$ are negatively associated, then
for any non-decreasing function $f_i, i \in [k]$, we have that:
$$
\E [ \prod_{i=1}^{k} f_i(X_i) ] \le \prod_{i=1}^{k} \E[ f_i(X_i) ].
$$
\end{proposition}
\begin{proof}
The proof follows directly from the definition of 
negatively associated variables.
\end{proof}

\subsection{Cache Analysis of In-place Permutation}
\label{sec:AnalInPlacePerm}
In this section we analyse the cache misses in a direct mapped cache
during $n$ rounds of Process ``in-place", introduced in 
Section~\ref{proc:Inplace}.
We derive a precise equation for the expected number of cache
misses and then give closed form upper and lower bounds on this
equation. We then derive upper and lower bounds assuming the keys
are drawn independently from a uniform distribution.

\subsubsection{Average case analysis}
We start by proving a theorem for the expected number of cache 
misses during $n$ rounds of Process ``in-place".

%
\begin{theorem}
\label{thm:THEOREM1}
The expected number $X$ of cache misses in $n$ rounds of Process ``in-place"
satisfies
%
$
n (p_c + p_d) \le X \le n(p_c+p_d) + k(1+1/B), \mbox{\it~where:}
$
\begin{eqnarray}
p_c &=& \sum_{i=1}^{k/B} P_i \left ( 1 - \sum_{m=0}^{\infty} P_i (1-P_i)^m
             \E_{\bar{\nu} \sim \varphi(m,\bar{b_i})}
\left[ \prod_{j=1}^{k}f(\nu_j) \right] \right ) \mbox{\rm~and}\nonumber \\
p_d &=& \frac{1}{B}  + \nonumber \\
    && \frac{B-1}{B} \sum_{i=1}^{k} 
          p_i \! \left ( 1-\!\! \sum_{m=0}^{\infty} p_i (1-p_i)^m
             \E_{\bar{\mu} \sim \varphi(m,\bar{a_i})}
                \left[ (1-g(\bar{\mu}))\prod_{j=1}^{k}f(\mu_j) \right]\right ) . \nonumber
\end{eqnarray}
\end{theorem}

\begin{proof}
We first analyse the miss rates for accesses to pointers $D_1, \ldots, D_k$.
Fix an $i$, $1\le i \le k$ and a $z \ge 1$. Let $\mu$ be the random
variable which denotes the number of rounds between accesses to locations
$d_{i,z}$ and $d_{i,z + 1}$ ($\mu = 0$ if these locations are
accessed in consecutive rounds).  
Figure~\ref{fig:ProcInplaceMiss} shows the other memory accesses
between accesses $z$ and $z+1$ to $D_i$.
Clearly, 
$\Pr[\mu = m] = p_i (1-p_i)^m, \mbox{\rm\ for $m = 0, 1, \ldots$}$.
Let $X_i$ denote the event that none of the memory
accesses in these $\mu$ rounds accesses the
cache block to which $d_{i,z}$ is mapped. We now fix an
%
%
integer $m \ge 0$ and calculate $\Pr[X_i | \mu = m]$.
Let $\bar{\mu}$ be a vector of random variables such that
for $1 \le j \le k$,  $\mu_j$ is the random variable which
denotes the number of accesses to $D_j$ in these $m$ rounds.
Clearly $\bar{\mu}$ is drawn from $\varphi(m,\bar{a_i})$ (note that
$D_i$ is not accessed in these $m$ rounds by definition).

\begin{figure}
\epsfxsize 4in
\centerline{\epsfbox{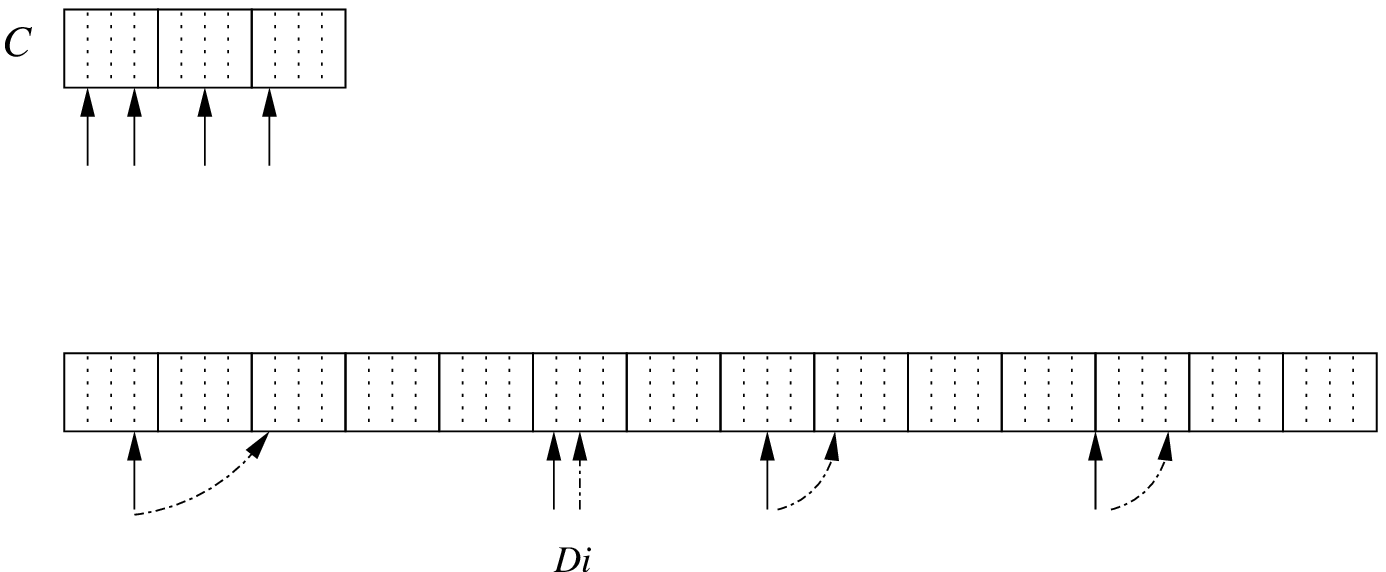}}
\caption{$m$ rounds of Process ``in-place". Between two accesses to $D_i$, there
are $m$ accesses to ``other" pointers, and $m+1$ accesses to ${\cal C}$.}
\label{fig:ProcInplaceMiss}
\end{figure}

Fix any vector $\bar{m}$, such that
$\Pr[\bar{\mu} = \bar{m}] \ne 0$, and let $\mu_j$ be 
the number of accesses to pointer $D_j$ in these
$m$ rounds.  Since $m_i$ must be zero, $f(m_i) = 1$, and for 
$j \ne i$, $f(m_j)$ is the probability that 
none of the $m_j$ locations accessed by $D_j$ 
in these $m$ rounds is mapped to the same cache 
block as location $d_{i,z}$~\cite{Mehlhorn2000,SC99}.
Similarly $g(\bar{m}) \cdot C$ is the number of
count blocks accessed in these rounds, and so
$1 - g(\bar{m})$ is the probability
that the cache block containing
$d_{i,z}$ does not conflict with the blocks from ${\cal C}$ which
were accessed in these $m$ rounds.   As the latter probability
is determined by the starting location of sequence $i$ and the
former probabilities 
by the starting location of sequences $j, j \ne i$, we conclude
that for a given configuration $\bar{m}$ of accesses,
the probability that the cache block containing $d_{i,z}$ is
not accessed in these $m$ rounds is
$(1-g(\bar{m}))\prod_{j=1}^{k} f(m_j)$.  Averaging over all
configurations $\bar{m}$, we get that
\begin{eqnarray}
\Pr[X_i \mid \mu = m] = E_{\bar{\mu} \sim \varphi(m,\bar{a_i})}
[ (1- g(\bar{\mu}))\prod_{j=1}^{k}f(\mu_j)].  
\end{eqnarray}
Finally we get,
\begin{eqnarray}
\Pr[X_i] &=& \sum_{m=0}^\infty \Pr[\mu = m] \Pr[X_i | \mu = m] \nonumber\\
         &=& \sum_{m=0}^\infty  p_i(1-p_i)^m
E_{\bar{\mu} \sim \varphi(m,\bar{a_i})}\
\left[ (1- g(\bar{\mu}))\prod_{j=1}^{k}f(\mu_j) \right].
\label{eq:pr_xi}
\end{eqnarray}

If $d_{i,z}$ is at a cache block boundary or if $X_i$ does not 
occur given that $d_{i,z}$ is not at a cache block boundary ($\Pr[X_i]$
does not change under this condition), then a 
cache miss will occur. The first access to a pointer is a cache miss. 
So other than for the first access, the probability $p_d$ of a cache 
miss for a pointer access is:
\begin{eqnarray}
p_d &=& \frac{1}{B} + \frac{B-1}{B} \sum_{i=1}^{k} p_i  ( 1-\Pr[X_i] ).
\label{eq:pd}
\end{eqnarray}

Including the first access misses, the expected number of cache misses 
for pointer accesses is at most
\begin{equation}
\sum_{i=1}^{k} 1+ (np_i-1)
\left(\left(\frac{B-1}{B} (1-\Pr[X_i]) \right)+\frac{1}{B}\right) \le
np_d +k.
\label{eq:ub_ptr}
\end{equation}

We now consider the probability of a cache miss for an access to 
a count array location.
It is convenient to partition ${\cal C}$ into
count blocks of $B$ locations each, where the $i$-th count block consists of
the locations $c_{(i-1)B+1}, \ldots, c_{iB}$, for $i = 1, \ldots, k/B$.
So $P_i$ is the probability of access to the $i$-th block.
We fix an $i \in \{1, \ldots, k/B\}$ and a $z\ge 1$.
Let $\nu$ be the random variable that denotes the number of
rounds between the $z$-th and $(z+1)$-st accesses to the $i$-th
count block. We have that $\Pr[\nu=m] = P_i(1-P_i)^m$, for
$m=0, 1, \ldots$.
Let $Y_i$ denote the event that none of the memory accesses in these
$m$ rounds accesses the cache block to which the $i$-th count block is mapped.
%

We now fix an integer $m \ge 0$ and calculate $\Pr[Y_i | \nu = m]$.
Let $\bar{\nu}$ be a vector of random variables such that
for $1 \le j \le k$,  $\nu_j$ is the random variable which
denotes the number of accesses to $D_j$ in these $m$ rounds.
Given that $k \le BC$ and assumption (c) 
mean that two blocks from ${\cal C}$ cannot conflict with each other. 
As the pointers $D_{(i-1)B+1}, \ldots, D_{iB}$ will not
be accessed between two successive accesses to count block $i$,
the probability of accessing pointer $D_j$ is given by $b^{i}_j$
and $\varphi(m,\bar{b_i})$ is the distribution for $\bar{\nu}$.
Arguing as above:

\begin{eqnarray}
\Pr[Y_i] &=& \sum_{m=0}^\infty \Pr[\nu = m] \Pr[Y_i | \nu = m] \nonumber\\
         &=& \sum_{m=0}^\infty P_i(1-P_i)^m
 E_{\bar{\nu} \sim \varphi(m,\bar{b_i})} \
 \left[\prod_{j=1}^{k}f(\nu_j) \right].
\label{eq:pr_yi}
 \end{eqnarray}
%

The first access to a count array block is a cache miss, for all other
accesses there is a cache miss if event $Y_i$ does not occur. So other
than for the first access, the probability $p_c$ of a cache miss for 
a count array access is:

\begin{equation}
p_c = \sum_{i=1}^{k/B} P_i  ( 1-\Pr[Y_i] ).
\label{eq:pc}
\end{equation}
%

Including the first access misses, the expected number of cache misses 
for count array accesses is at most
\begin{equation}
\sum_{i=1}^{k/B} 1 + (nP_i-1) (1-\Pr[Y_i]) \le np_c +k/B.
\label{eq:ub_count}
\end{equation}


Plugging in the values from Eq.~\ref{eq:pr_xi} into Eq.~\ref{eq:ub_ptr} 
and from Eq.~\ref{eq:pr_yi} into Eq.~\ref{eq:ub_count} 
we get the upper bound on $X$, 
the expected number of cache misses in the processes.
The lower bound in Theorem~\ref{thm:THEOREM1} is obvious.

\end{proof}

\subsubsection{Upper bound}
We now prove a theorem on the upper bound to the expected
number of cache misses during $n$ rounds of Process ``in-place".

\begin{theorem}
\label{thm:THEOREM2}
The expected number of cache misses in $n$ rounds of Process ``in-place"
is at most $n(p_d + p_c) + k(1+1/B)$, where:
\begin{eqnarray*}
p_d &\le& \frac{1}{B} + \frac{k}{BC} + 
\frac{B-1}{BC} \sum_{i=1}^{k} \left(
\sum_{j=1}^{k/B}\frac{p_iP_j}{p_i+P_j} + 
\frac{B-1}{B}\sum_{j=1}^{k}\frac{p_ip_j}{p_i+p_j} \right ), \\
p_c &\le& \frac{k}{B^2C} + \frac{B-1}{BC}\sum_{i=1}^{k/B} 
\sum_{j=1}^{k}\frac{P_ip_j}{P_i+p_j}.
\end{eqnarray*}
\end{theorem}
\begin{proof}
In the proof we derive lower bounds for $\Pr[X_i]$ and $\Pr[Y_i]$ 
and use these to derive the upper bounds on $p_d$ and $p_c$.

Again, we consider a fixed $i$ and consider the event $X_i$ 
defined in the proof of Theorem~\ref{thm:THEOREM1}.
We now obtain a lower bound on $\Pr[X_i]$.

\subsubsection*{Lower bound on $\Pr[X_i]$}
Letting $\Gamma(x)=1-f(x)$ and using Proposition~\ref{prop:ProdSum}
we can rewrite Eq.~\ref{eq:pr_xi} as:
\begin{equation}
\Pr[X_i] \ge \sum_{m=0}^\infty \Pr[\mu = m]
                 \E_{\bar{\mu} \sim \varphi(m,\bar{a_i})}
                  \left [ 1  - g(\bar{\mu})
                             - \sum_{j=1}^k \Gamma(\mu_j) \right].
\label{eq:xiub_part2}
\end{equation}

We know that the $j$-th count block contributes $1/C$ to 
$g(\bar{\mu})$ if there is an access to that block and
$\Pr[j\mbox{\rm-th count block accessed}|\mu=m]=1-(1-c^{i}_j)^m$,
where $c^{i}_j=\frac{P_j}{1-p_i}$.
So we have that, 
$$
\E_{\bar{\mu} \sim \varphi(m,\bar{a_i})}[g(\bar{\mu})] =
\sum_{j=1}^{k/B} \frac{1}{C} (1-(1-c^{i}_j)^m),
$$
and we get,
\begin{eqnarray}
\sum_{m=0}^\infty \Pr[\mu = m]
\E_{\bar{\mu} \sim \varphi(m,\bar{a_i})}[g(\bar{\mu})]
  &=& \sum_{m=0}^{\infty} p_i(1-p_i)^m
      \sum_{j=1}^{k/B} \frac{1}{C}(1-(1-c^{i}_j)^m)\nonumber\\
  &=& \frac{1}{C}\sum_{j=1}^{k/B} \sum_{m=0}^{\infty} p_i(1-p_i)^m  (1-(1-c^{i}_j)^m), \nonumber
\end{eqnarray}
and using Proposition~\ref{prop:prop3} we get,
\begin{eqnarray}
\sum_{m=0}^\infty \Pr[\mu = m]
\E_{\bar{\mu} \sim \varphi(m,\bar{a_i})}[g(\bar{\mu})]
  = \frac{1}{C}\sum_{j=1}^{k/B} \frac{P_j}{p_i+P_j}.
\label{eq:gmu_ub}
\end{eqnarray}

We now evaluate 
$$\sum_{m=0}^\infty \Pr[\mu=m] \sum_{j=1}^{k}
\E_{\bar{\mu} \sim \varphi(m,\bar{a_i})}[\Gamma(\mu_j)].$$
Our approach is to first fix $j$ and evaluate 
$\E_{\bar{\mu} \sim \varphi(m,\bar{a_i})}[\Gamma(\mu_j)]$.
For $m \le BC$, we know that 
$$
\E_{\bar{\mu} \sim \varphi(m,\bar{a_i})}[\Gamma(\mu_j)] = 
\sum_{l=0}^{m} \Pr[\mu_j=l] \frac{l+B-1}{BC} -\Pr[\mu_j=0] \frac{B-1}{BC}.
$$
The last term is due to the fact that $\Gamma(x)$ is discontinuous and 
$\Gamma(0) = 0$. Similarly for $m>BC$ we know that
\begin{eqnarray}
\E_{\bar{\mu} \sim \varphi(m,\bar{a_i})}[\Gamma(\mu_j)] &=&
\sum_{l=0}^{m} \Pr[\mu_j=l] \frac{l+B-1}{BC} - \Pr[\mu_j=0] \frac{B-1}{BC} \nonumber \\
&& -\sum_{l=BC-B+1}^{m} \Pr[\mu_j=l] \left(  \frac{l+B-1}{BC} -1 \right). \nonumber
\end{eqnarray}
The last term is due to the fact that $\Gamma(x)=1$ for $x \ge BC-B+1$.
If we drop this last term when $m>BC$, we get that for all $m$
\begin{eqnarray}
\E_{\bar{\mu} \sim \varphi(m,\bar{a_i})}[\Gamma(\mu_j)] & \le &
\frac{1}{BC} \left[ \sum_{l=0}^{m} \Pr[\mu_j=l] l + 
(B-1) (1-\Pr[\mu_j=0]) \right]. \nonumber
\end{eqnarray}
The summation term is the expected value of the random variable with
the binomial distribution $b(l;m,a^{i}_{j})$. So we get that
\begin{equation}
\label{eq:GivenmGammamuj}
\E_{\bar{\mu} \sim \varphi(m,\bar{a_i})}[\Gamma(\mu_j)] 
        \le \frac{1}{BC} \left [ m a^{i}_j + 
                 (B-1) \left (1 - \left(1-a^{i}_j \right)^m \right) \right].
\end{equation}

We now evaluate 
$\sum_{m=0}^\infty \Pr[\mu=m] \sum_{j=1}^{k}
\E_{\bar{\mu} \sim \varphi(m,\bar{a_i})}[\Gamma(\mu_j)]$ as
\begin{eqnarray}
\lefteqn{\sum_{m=0}^\infty \Pr[\mu=m] \sum_{j=1}^{k}
\E_{\bar{\mu} \sim \varphi(m,\bar{a_i})}[\Gamma(\mu_j)] }\nonumber \\
   &\le& \sum_{m=0}^\infty \Pr[\mu=m] 
		 \sum_{j=1}^{k} \frac{1}{BC} \left [ m a^{i}_j +
                 (B-1) \left (1 - \left(1-a^{i}_j \right)^m \right) \right]. \nonumber
\end{eqnarray}
Since $\sum_{j=1}^{k}  m a^{i}_j = m$, we get
$\sum_{m=0}^\infty \Pr[\mu=m] \sum_{j=1}^{k}  m a^{i}_j = \frac{1}{p_i}-1$ by
an application of Proposition~\ref{prop:GeoSeriesc}(c).
By applying Proposition~\ref{prop:prop3} we get that 
$\sum_{j=1}^{k} \sum_{m=0}^\infty \Pr[\mu=m] 
(B-1)(1 - (1-a^{i}_j )^m )= (B-1)\sum_{j=1}^{k}\frac{p_j}{p_i+p+j}$. So we
get:
\begin{eqnarray}
\sum_{m=0}^\infty \Pr[\mu=m] \sum_{j=1}^{k}
\E_{\bar{\mu} \sim \varphi(m,\bar{a_i})}[\Gamma(\mu_j)] 
 &\le& \frac{1}{BC} \left( \frac{1}{p_i} +
                           (B-1) \sum_{j=1}^{k} \frac{p_j}{p_i+p_j}\right).\nonumber\\
\label{eq:Gammamuj_ub}
\end{eqnarray}

Substituting Eq.~\ref{eq:gmu_ub} and~\ref{eq:Gammamuj_ub} in 
Eq.~\ref{eq:xiub_part2} we obtain the following lower bound for $\Pr[X_i]$
\begin{eqnarray}
\Pr[X_i] &\ge& 1 - 
\frac{1}{C}\sum_{j=1}^{k/B}\frac{P_j}{p_i+P_j} - \frac{1}{BC}
  \left ( \frac{1}{p_i} + (B-1)\sum_{j=1}^{k}\frac{p_j}{p_i+p_j} \right ).
\label{eq:xiub_final}
\end{eqnarray}

\subsubsection*{Upper bound on $p_d$}
Finally, substituting $\Pr[X_i]$ from Eq.~\ref{eq:xiub_final} in Eq.~\ref{eq:pd} 
we get: 
\begin{eqnarray}
p_d &\le& \frac{1}{B} +  \nonumber\\
   && \frac{B-1}{B} \sum_{i=1}^{k} p_i 
   \left ( 
   \frac{1}{C}\sum_{j=1}^{k/B} \frac{P_j}{p_i+P_j} +
   \frac{1}{BC}\left ( \frac{1}{p_i} + 
   (B-1)\sum_{j=1}^{k} \frac{p_j}{p_i+p_j} \right) \right) \nonumber\\
   &=& \frac{1}{B} + \frac{(B-1)k}{B^2C} +
   \frac{B-1}{BC} \sum_{i=1}^{k} \sum_{j=1}^{k/B} \frac{p_iP_j}{p_i+P_j} +
   \frac{(B-1)^2}{B^2C} \sum_{i=1}^{k} \sum_{j=1}^{k} \frac{p_ip_j}{p_i+p_j} \nonumber \\
   &\le& \frac{1}{B} + \frac{k}{BC} + 
\frac{B-1}{BC} \sum_{i=1}^{k} \left(
\sum_{j=1}^{k/B}\frac{p_iP_j}{p_i+P_j} + 
\frac{B-1}{B}\sum_{j=1}^{k}\frac{p_ip_j}{p_i+p_j} \right ). \nonumber
\end{eqnarray}

We can evaluate $p_c$ using a very similar approach, as sketched
out now. 
We again consider a fixed $i$ and consider the event $Y_i$ 
defined in the proof of Theorem~\ref{thm:THEOREM1}.
We now obtain a lower bound on $\Pr[Y_i]$.

\subsubsection*{Lower bound on $\Pr[Y_i]$}
Again letting $\Gamma(x) = 1 - f(x)$ and using Proposition~\ref{prop:ProdSum},
we can rewrite Eq.~\ref{eq:pr_yi} as:
\begin{equation}
\Pr[Y_i] \ge \sum_{m=0}^\infty \Pr[\nu = m]
                 \E_{\bar{\nu} \sim \varphi(m,\bar{b_i})}
                  \left [ 1  - \sum_{j=1}^k \Gamma(\nu_j) \right]
\label{eq:yiub_part2}
\end{equation}
Arguing as for the derivation of Eq.~\ref{eq:GivenmGammamuj}, we get
$$
\E_{\bar{\nu} \sim \varphi(m,\bar{b_i})}[\Gamma(\nu_j)] 
        \le \frac{1}{BC} \left [ m b^{i}_j + 
                 (B-1) \left (1 - \left(1-b^{i}_j \right)^m \right) \right].
$$
Then arguing as for the derivation of Eq.~\ref{eq:Gammamuj_ub}, we get
\begin{eqnarray}
\sum_{m=0}^\infty \Pr[\nu=m] \sum_{j=1}^{k}
\E_{\bar{\nu} \sim \varphi(m,\bar{b_i})}[\Gamma(\nu_j)] 
 \le \frac{1}{BC} \left( \frac{1}{P_i} +
                           (B-1) \sum_{j=1}^{k} \frac{p_j}{P_i+p_j}\right).\nonumber
\label{eq:Gammanuj_ub}
\end{eqnarray}
Substituting this into Eq.~\ref{eq:yiub_part2}, we get:
\begin{eqnarray}
\Pr[Y_i] &\ge& 1 -
\frac{1}{BC} \left(\frac{1}{P_i}+(B-1)\sum_{j=1}^{k}\frac{p_j}{P_i+p_j} \right).
\label{eq:yiub_final} 
\end{eqnarray}

\subsubsection*{Upper bound on $p_c$}
Substituting $\Pr[Y_i]$ from Eq.~\ref{eq:yiub_final} in Eq.~\ref{eq:pc} we get
\begin{eqnarray}
p_c &\le& \sum_{i=1}^{k/B} P_i \frac{1}{BC} \left( \frac{1}{P_i} + 
         (B-1)\sum_{i=1}^{k} \frac{p_j}{P_i+p_j} \right) \nonumber\\
    &=& \frac{k}{B^2C} + \frac{B-1}{BC}\sum_{i=1}^{k/B} 
        \sum_{j=1}^{k}\frac{P_ip_j}{P_i+p_j}. \nonumber
\end{eqnarray}

\end{proof}
This proves the upper bound for the equation in Theorem~\ref{thm:THEOREM1}. 
We now prove a lower bound on that equation.

%
%

\subsubsection{Lower bound}
%
\begin{theorem}
\label{thm:THEOREM3}
When $p_i \ge 1/C$ then
the expected number of cache misses in $n$ rounds of Process ``in-place"
is at least $np_d + k$, where:
\begin{eqnarray*}
p_d \ge \frac{1}{B} &+& \frac{k(2C-k)}{2C^2} + \frac{k(k-3C)}{2BC^2} - \frac{1}{2BC} - \frac{k}{2B^2C} \\
&+&
\frac{B(k-C) +2C-3k}{BC^2} \sum_{i=1}^{k} \sum_{j=1}^{k} \frac{(p_i)^2}{p_i+p_j}
\\
&+& \frac{(B-1)^2}{B^3C^2} \sum_{i=1}^{k} p_i \left [ 
\sum_{j=1}^{k} \frac{p_i (1-p_i-p_j)}{(p_i+p_j)^2}
- \frac{B-1}{2}\sum_{j=1}^{k} \sum_{l=1}^{k} \frac{p_i}{p_i + p_j +p_l - p_j p_l} \right] - O\left(e^{-B}\right).
\end{eqnarray*}
\end{theorem}
\begin{proof}
We again consider a fixed $i$ and consider the event $X_i$ 
defined in the proof of Theorem~\ref{thm:THEOREM1}. 
Let $\bar{\mu}$ be as defined in the proof of Theorem~\ref{thm:THEOREM1}.
We now obtain an upper bound on $\Pr[X_i]$. 

\subsubsection*{Upper bound on $\Pr[X_i]$}
In~\cite{DPR96} it is
shown that the variables $\mu_j$ are negatively associated
\cite{JP83}. Noting that $f(x)$ is a non-increasing function of $x$, then
using Proposition~\ref{prop:NegAssoc} we have that:
$$
\E_{\bar{\mu} \sim \varphi(m,\bar{a_i})}[\prod_{j=1}^k f(\mu_j)] \le \
\prod_{j=1}^k \E_{\bar{\mu} \sim \varphi(m,\bar{a_i})}[f(\mu_j)].
$$
So we can re-write Eq.~\ref{eq:pr_xi} as:
\begin{eqnarray}
\label{eq:lb_eq1}
\Pr[X_i]   \le  \sum_{m=0}^{BC-B} \Pr[\mu = m] \prod_{j=1}^{k}
\E_{\bar{\mu} \sim \varphi(m,\bar{a_i})}[f(\mu_j)] +
\sum_{m=BC-B+1}^{\infty} \Pr[\mu = m].  \nonumber \\
\end{eqnarray}

We first bound the last term.  We know that
\begin{eqnarray*}
\sum_{m=BC-B+1}^{\infty} \Pr[\mu = m] &=&
(1-p_i)^{BC-B+1} \sum_{m=0}^{\infty} \Pr[\mu = m] \\
&=& (1-p_i)^{BC-B+1}.
\end{eqnarray*}
Using Proposition~\ref{prop:prop6} we get that
$(1-p_i)^{BC-B+1} \le e^{-(BC-B+1)p_i}$. 
Assuming $p_i \ge 1/C$ the last term is at most $O(e^{-B})$.

We now bound the first term in Eq.~\ref{eq:lb_eq1}. 
We use an approach similar to the derivation of
Eq.~\ref{eq:GivenmGammamuj} and since $\mu \le BC-B$, so $\mu_j \le BC-B$,
we don't have to drop any terms in the simplification, so we get that:
$$
\E_{\bar{\mu} \sim \varphi(m,\bar{a_i})}[f(\mu_j)] = 
1 - \frac{1}{BC}(m a^{i}_j  + (B-1)(1-(1-a^{i}_j)^m)).
$$
%
Letting $t_j(m) = \frac{1}{BC}(m a^{i}_j  + (B-1)(1-(1-a^{i}_j)^m))$ and
using Proposition~\ref{prop:Taylora}(a) we get that $e^{-\sum_j {t_j(m)}} \ge \prod_j (1-t_j(m))$. 
So we have that 
\begin{eqnarray*}
\Pr[X_i] &\le& \sum_{m=0}^{BC-B} \Pr[\mu = m]
           e^{\frac{-1}{BC} \sum_{j=1}^{k} \left(m a^{i}_j+(B-1) (1-(1-a^{i}_j)^m)) \right)}+ O(e^{-B}).\\
         &\le& \sum_{m=0}^{BC-B} \Pr[\mu = m]
           e^{\frac{-1}{BC} \left(m+(B-1) (k - \sum_{j=1}^{k}(1-a^{i}_j)^m) \right)}+ O(e^{-B}).
\end{eqnarray*}
Using Proposition~\ref{prop:Taylorb}(b) and letting 
$\beta_j = (1-a^{i}_j)$ we get that:
\begin{eqnarray}
\Pr[X_i] &\le& \sum_{m=0}^{BC-B} \Pr[\mu = m] \left [ 1 - 
       \frac{(B-1)k}{BC} - \frac{((B-1)k)^2}{2(BC)^2} - \frac{m^2}{2(BC)^2} \nonumber \right. \\
    &&-   \frac{1}{BC} \left(m - (B-1)\sum_{j=1}^{k}\beta_j^m) \right) 
		 - \frac{(B-1)}{2(BC)^2}  \left( 2m\sum_{j=1}^{k}\beta_j^m + 2 (B-1) k \sum_{j=1}^{k}\beta_j^m  \right) \nonumber\\
&&+\left.  \frac{(B-1)}{2(BC)^2} \left( 2mk  + (B-1) \sum_{j=1}^{k} \sum_{l=1}^{k}\beta_j^m \beta_l^m \right) 
\right] + O(e^{-B}) .
\label{eq:lb_prop6}
\end{eqnarray}

We now evaluate the terms in Eq.~\ref{eq:lb_prop6} assuming that $p_i \ge 1/C$,
so $k\le C$. For the simplifications of the subtractive terms we use
the fact that $e^{-p_i (BC - B +1)} \le e^{-B}$.

%
Since $p_i \ge 1/C$,
$( 1/p_i + BC-B) / (BC) = O(1)$, 
so using Proposition~\ref{prop:SumProb0toya}(a), we get that
\begin{eqnarray}
\label{eq:lb_prop6_term1}
\sum_{m=0}^{BC-B} \Pr[\mu = m] \frac{m}{BC} =
\frac{1}{BC}\left(\frac{1}{p_i} -1\right) - O(e^{-B}).
\end{eqnarray}

%
Since $p_i \ge 1/C$,
$( B-1)k / (BC) < 1$,
so using Proposition~\ref{prop:SumProb0toyb}(b), we get that
\begin{eqnarray}
\label{eq:lb_prop6_term2}
\sum_{m=0}^{\alpha-1} \Pr[\mu = m] \frac{(B-1)k}{BC} =
          \frac{(B-1)k}{BC} - O(e^{-B}).
\end{eqnarray}

We now evaluate the term
\begin{eqnarray}
\lefteqn{\frac{(B-1)}{(BC)^2}\sum_{m=0}^{\alpha-1} \Pr[\mu = m] m \sum_{j=1}^{k}\beta_j^m  }\nonumber\\
&=& \frac{(B-1)}{(BC)^2}\sum_{j=1}^{k}\sum_{m=0}^{\infty} \Pr[\mu = m] m \beta_j^m \nonumber \\
&& - (1-p_i)^{\alpha} \frac{(B-1)}{(BC)^2}\sum_{m=0}^{\infty} \Pr[\mu = m] (m+\alpha) \sum_{j=1}^{k}\beta_j^m \nonumber\\
%
%
&=& \frac{(B-1)}{(BC)^2}\sum_{j=1}^{k} \frac{p_i (1-p_i-p_j)}{(p_i+p_j)^2} \nonumber \\
&& - (1-p_i)^{\alpha} \frac{(B-1)}{(BC)^2}\left(\sum_{j=1}^{k} \frac{p_i (1-p_i-p_j)}{(p_i+p_j)^2} 
+\sum_{j=1}^{k} \frac{\alpha p_i}{p_i+p_j} \right)\nonumber\\
&=& \frac{(B-1)}{(BC)^2}\sum_{j=1}^{k} \frac{p_i (1-p_i-p_j)}{(p_i+p_j)^2} - O(e^{-B}).
\label{eq:lb_prop6_term3}
\end{eqnarray}
The last simplification is due to $p_i \ge 1/C$ and $k \le C$, so 
$\sum_{1 \le j \le k} (p_i (1-p_i-p_j))/(p_i+p_j)^2 \le kC \le C^2$ and 
$\sum_{1 \le j \le k} (\alpha p_i)/(p_i+p_j) \le kCB \le BC^2.$ \\

Substituting back $(1-a^{i}_j) = \beta_j$ and using Proposition~\ref{prop:prop3} we get that
\begin{eqnarray}
\lefteqn{\frac{(B-1)^2k}{(BC)^2}\sum_{m=0}^{BC-B} \Pr[\mu = m] \sum_{j=1}^{k}\beta_j^m } \nonumber \\
&&= \frac{(B-1)^2k}{(BC)^2}\sum_{j=1}^{k}\frac{p_i}{p_i+p_j} - (1-p_i)^{\alpha} \frac{(B-1)^2k}{(BC)^2}\sum_{j=1}^{k}\frac{p_i}{p_i+p_j} \nonumber \\
&&= \frac{(B-1)^2k}{(BC)^2}\sum_{j=1}^{k}\frac{p_i}{p_i+p_j} - O(e^{-B}).
\label{eq:lb_prop6_term4}
\end{eqnarray}
The last step used $\sum_{j=1}^{k}p_i/(p_i+p_j) \le k$ and $((B-1)k)^2/(BC)^2) < 1$.

We now evaluate the additive terms, starting with
\begin{eqnarray}
\sum_{m=0}^{BC-B} \Pr[\mu = m] \frac{mk(B-1)}{(BC)^2} \le \frac{(B-1)k}{(BC)^2} \left(\frac{1}{p_i} -1 \right).
\label{eq:lb_prop6_term5}
\end{eqnarray}
%

Since $m<BC$ we now get that:
\begin{eqnarray}
 \sum_{m=0}^{BC-B} \Pr[\mu = m] \frac{m^2}{2(BC)^2} \le \frac{1}{2BC}\left(\frac{1}{p_i} -1\right).
 \label{eq:lb_prop6_term6}
\end{eqnarray}
Substituting back $(1-a^{i}_j) = \beta_j$ and using Proposition~\ref{prop:prop3} we get that:
\begin{eqnarray}
\frac{B-1}{BC} \sum_{m=0}^{BC-B} \Pr[\mu = m] \sum_{j=0}^{k}{\beta_j}^m \le 
\frac{B-1}{BC} \sum_{j=1}^{k} \frac{p_i}{p_i+p_j}.
\label{eq:lb_prop6_term7}
\end{eqnarray}
Finally we evaluate 
$\sum_{m=0}^{\alpha-1} \Pr[\mu = m] \sum_{j=1}^{k} \sum_{l=1}^{k} {\beta_j}^m {\beta_l}^m$,
by first evaluating
\begin{eqnarray*}
(1-p_i) \beta_j \beta_l &=& (1-p_i) ( 1 - a^{i}_j ) ( 1 - a^{i}_l ) \nonumber \\
                      &=& \frac{1 - 2p_i - p_j - p_l + p_i(p_i + p_j +p_l) + p_jp_l}{(1-p_i)}.
\end{eqnarray*}
Using this result and Proposition~\ref{prop:GeoSeriesa}(a), we get that
\begin{eqnarray*}
\lefteqn{\sum_{m=0}^{\alpha-1} \Pr[\mu = m] {\beta_j}^m {\beta_l}^m} \nonumber\\
&&\le \sum_{m=0}^{\infty} 
       p_i \left(\frac{1 - 2p_i - p_j - p_l + p_i(p_i + p_j +p_l) + p_jp_l}{(1-p_i)} \right)^m \nonumber \\
&&= \frac{p_i}{p_i + p_j +p_l - p_j p_l/(1-p_i)} \nonumber \\
&&\le \frac{p_i}{p_i + p_j +p_l - p_jp_l}.
\end{eqnarray*}
So we get that 
\begin{eqnarray}
\lefteqn{\frac{(B-1)^2}{2(BC)^2}\sum_{m=0}^{\alpha-1} 
             \Pr[\mu = m] \sum_{j=1}^{k} \sum_{l=1}^{k} {\beta_j}^m {\beta_l}^m} \nonumber \\
&& \le \frac{(B-1)^2}{2(BC)^2}\sum_{j=1}^{k} \sum_{l=1}^{k} \frac{p_i}{p_i + p_j +p_l - p_j p_l}.
\label{eq:lb_prop6_term8}
\end{eqnarray}

\subsubsection*{Lower bound on $p_d$}
Plugging Eqs.~\ref{eq:lb_prop6_term1}~$,..,$~\ref{eq:lb_prop6_term8} into Eq.~\ref{eq:lb_prop6} we get that:
\begin{eqnarray}
\Pr[X_i] &\le& 1  
- \frac{1}{BC}\left(\frac{1}{p_i} -1\right) 
- \frac{(B-1)k}{BC}
- \frac{(B-1)}{(BC)^2}\sum_{j=1}^{k} \frac{p_i (1-p_i-p_j)}{(p_i+p_j)^2} \nonumber \\
&-& \frac{(B-1)^2k}{(BC)^2}\sum_{j=1}^{k}\frac{p_i}{p_i+p_j}
+ \frac{(B-1)k}{(BC)^2} \left(\frac{1}{p_i} -1 \right)
+ \frac{1}{2BC}\left(\frac{1}{p_i} -1\right) \nonumber\\
&+& \frac{B-1}{BC} \sum_{j=1}^{k} \frac{p_i}{p_i+p_j} 
+ \frac{(B-1)^2}{2(BC)^2}\sum_{j=1}^{k} \sum_{l=1}^{k} \frac{p_i}{p_i + p_j +p_l - p_j p_l} \nonumber \\
&+& \frac{((B-1)k)^2}{2(BC)^2} + O(e^{-B}).
\label{eq:lb_pxi}
\end{eqnarray}

Plugging Eq.~\ref{eq:lb_pxi} into Eq.~\ref{eq:pd} we get:
\begin{eqnarray*}
p_d \ge \frac{1}{B} + \frac{B-1}{B} \sum_{i=1}^{k} p_i  &&\left[ 
\left(\frac{1}{p_i} -1\right) \frac{1}{2BC} \left (1 - \frac{2(B-1)k}{BC} \right)  \right.\\
&+& \sum_{j=1}^{k} \frac{p_i}{p_i+p_j} \frac{B-1}{BC} \left(  \frac{(B-1)k}{BC} - 1 \right) \\
&+& \frac{(B-1)k}{BC} \left(1 - \frac{(B-1)k}{2BC} \right) \\
&+& \frac{(B-1)}{(BC)^2} \sum_{j=1}^{k} \frac{p_i (1-p_i-p_j)}{(p_i+p_j)^2} \\
&-& \left. \frac{(B-1)^2}{2(BC)^2}\sum_{j=1}^{k} \sum_{l=1}^{k} \frac{p_i}{p_i + p_j +p_l - p_j p_l} 
- O\left(e^{-B}\right) \right].
\end{eqnarray*}

Simplifying further and
using $\sum_{i=1}^{k} \sum_{j=1}^{k} {p_i}^2 / (p_i+p_j) < k$
and $\sum_{i=1}^{k} p_i( 1/p_i - 1) = k-1$, we get: 
\begin{eqnarray*}
p_d \ge \frac{1}{B} &+& \frac{k(2C-k)}{2C^2} + \frac{k(k-3C)}{2BC^2} - \frac{1}{2BC} - \frac{k}{2B^2C} \\
&+&
\frac{B(k-C) +2C-3k}{BC^2} \sum_{i=1}^{k} \sum_{j=1}^{k} \frac{(p_i)^2}{p_i+p_j}
\\
&+& \frac{(B-1)^2}{B^3C^2} \sum_{i=1}^{k} p_i \left [ 
\sum_{j=1}^{k} \frac{p_i (1-p_i-p_j)}{(p_i+p_j)^2}
- \frac{B-1}{2}\sum_{j=1}^{k} \sum_{l=1}^{k} \frac{p_i}{p_i + p_j +p_l - p_j p_l} \right] - O\left(e^{-B}\right).
\end{eqnarray*}

\end{proof}


\subsubsection{Upper and lower bounds for uniformly random data}

Using the upper and lower bound Theorems just proven for general probability
distributions, we now derive Corollaries for upper and lower bounds for 
uniform distribution. 

\begin{corollary}
\label{cor:UB_Cor1}
If $p_1 = \ldots = p_k = 1/k$ then the number of cache misses in $n$ rounds
of Process ``in-place" is at most :
$$
n \left( \frac{1}{B} + \frac{k(B+5)}{2BC} + \frac{k}{B^2C}\right)
+ k\left(1+\frac{1}{B}\right).
$$
\end{corollary}
\begin{proof}
Since $P_i$ in $p_c$ and $P_j$ in $p_d$ are both $B/k$ in the equations in 
Theorem~\ref{thm:THEOREM2}, we get that:
\begin{eqnarray}
p_d+p_c &\le& 
\frac{1}{B} + \frac{2(B-1)}{BC}\frac{k^2}{B}\frac{B/k}{B+1} 
+\frac{(B-1)^2}{B^2C}k^2\frac{1/k}{2}
+ \frac{k}{B^2C}+\frac{k}{BC} \nonumber\\
&=& \frac{1}{B} + \frac{2(B-1)}{BC}\frac{k}{B+1} 
+\frac{(B-1)^2}{B^2C}\frac{k}{2}
+ \frac{k}{B^2C}+\frac{k}{BC} \nonumber\\
&\le& \frac{1}{B} 
+\frac{k}{C} \left[ \frac{3}{B}+\frac{B-1}{2B} \right]
+ \frac{k}{B^2C} \nonumber\\
&=& \frac{1}{B} + \frac{k(B+5)}{2BC} + \frac{k}{B^2C}. \nonumber
\end{eqnarray}
\end{proof}

\begin{remark}
As we will see later, Process ``in-place" models the permute phase of 
distribution sorting and Corollary~\ref{cor:UB_Cor1} shows that one pass of
uniform distribution sorting incurs $O(n/B)$ cache misses if and only
if $k=O(C/B)$.
\end{remark}

The following corollary is from the lower bound result in Theorem~\ref{thm:THEOREM3}.
\begin{corollary}
\label{cor:LB_Cor2}
If $p_1 = \ldots = p_k = 1/k$ then the number of cache misses in $n$ rounds
of Process ``in-place" is at least:
\begin{eqnarray*}
 k + \frac{n}{B} +
n \left [ \frac{k}{2C} - \frac{k^2}{BC^2} - \frac{k+1}{2BC} 
- \frac{k}{2B^2C}
+\frac{(B-1)^2}{12B^3C^2}\left( k^2(5-2B) -7k + 2 \right) \right]. \nonumber
\end{eqnarray*}
\end{corollary}
\begin{proof}
Plugging $p_i=1/k$ in the equation in Theorem~\ref{thm:THEOREM3} we get that:
\begin{eqnarray}
p_d \ge \frac{1}{B} &+& \frac{k(2C-k)}{2C^2} + \frac{k(k-3C)}{2BC^2} - \frac{1}{2BC} - \frac{k}{2B^2C} \nonumber \\
&+&\frac{B(k-C) +2C-3k}{BC^2} \frac{k}{2}
 + \frac{(B-1)^2}{B^3C^2} \left[ \frac{(k-2)k}{4}  
 -\frac{B-1}{2} \left( \frac{k^3}{3k-1} \right) \right] \nonumber \\
\ge \frac{1}{B} &+&
\frac{k}{2C} - \frac{k^2}{BC^2} - \frac{k+1}{2BC} 
- \frac{k}{2B^2C}
+\frac{(B-1)^2}{12B^3C^2}\left( k^2(5-2B) -7k + 2 \right). \nonumber
\end{eqnarray}
\end{proof}

\begin{remark}
From Corollary~\ref{cor:UB_Cor1} we have that for uniformly random data and
$k=\alpha C$, where $\alpha\le 1$, other than for small values of $B$, 
the upper bound for the number of cache misses in $n$ round is roughly
$$\frac{\alpha n}{2},$$
and from Corollary~\ref{cor:LB_Cor2} we have that for uniformly random 
data and $k=\alpha C$, where $\alpha\le 1$, other than for small values of $B$, 
the lower bound for the number of cache misses in $n$ round is roughly
$$n \left (\frac{\alpha}{2} - \frac{{\alpha}^2}{6} \right ).$$

The ratio between the upper and lower bound is $3/(3-\alpha)$.
So we have that for uniformly random data the lower bound is within
a factor of about $3/2$ of the upper bound when $k \le C$ and is much 
closer when $k \ll C$.

\end{remark}

\subsection{Cache Analysis of Out-of-place Permutation}
\label{sec:AnalOutofPlacePerm}
In this section we analyse the cache misses in a direct mapped cache
during $n$ rounds of Process ``out-of-place", introduced in 
Section~\ref{proc:Outplace}.
We derive a precise equation for the expected number of cache misses
and closed-form upper and lower bounds.
During the analysis we re-use 
$k, D_i, c_i, {\cal C}, p_i, P_i, \bar{a}, \bar{b}, f(x)$
and $g(m)$ introduced in Section~\ref{sec:AnalInPlacePerm}. 

\subsubsection{Average case analysis}
We start by proving a theorem for the expected number of cache 
misses during $n$ rounds of Process ``out-of-place".

\begin{theorem}
\label{thm:OPP_THEOREM1}
The expected number $X$ of cache misses in $n$ rounds of Process ``out-of-place" is
$
n (p_c + p_d + p_s) \le X \le n(p_c+p_d +p_s) + k(1+1/B) + 1, \mbox{\it~where:}
$
\begin{eqnarray}
p_c &=& \sum_{i=1}^{k/B} P_i \left ( 1 - \sum_{m=0}^{\infty} P_i (1-P_i)^m
             \E_{\bar{\nu} \sim \varphi(m,\bar{b_i})}
\left[ f(m + 1) \prod_{j=1}^{k}f(\nu_j) \right] \right ), \nonumber \\
p_d &=& \frac{1}{B} +\frac{B-1}{B} \sum_{i=1}^{k} p_i \nonumber\\ 
&&           \left ( 1-\!\! \sum_{m=0}^{\infty} p_i (1-p_i)^m
             \E_{\bar{\mu} \sim \varphi(m,\bar{a_i})}
                \left[ (1-g(\bar{\mu})) f(m + 1) \prod_{j=1}^{k}f(\mu_j) \right]\right ), \nonumber\\
p_s &=& \frac{1}{B} + \frac{B-1}{B} \left(1 - \left( 1 - \frac{1}{C} \right)^2 \right). \nonumber
\end{eqnarray}
\end{theorem}

\begin{proof}
We first analyse the miss rates for accesses to pointers $D_1, \ldots, D_k$.
Fix an $i$, $1\le i \le k$ and a $z \ge 1$ and consider
the probability of a miss between access $z$ and $z+1$ to pointer $D_i$. 
We define $\mu, \mu_j, m, \bar{\mu},
\bar{m}, d_{i,z}$, $\varphi(m,\bar{a_i})$ and $X_i$ as in the proof of 
Theorem~\ref{thm:THEOREM1}.
Again $f(m_j)$ is the probability that 
none of the $m_j$ locations accessed by $D_j$ 
in $m$ rounds is mapped to the same cache 
block as location $d_{i,z}$.
Similarly $g(\bar{m}) \cdot C$ is the number of
count blocks accessed in $m$ rounds, and so
$1 - g(\bar{m})$ is the probability
that the cache block containing
$d_{i,z}$ does not conflict with the blocks from ${\cal C}$ which
were accessed in these $m$ rounds.   
We also have accesses to $m+1$ contiguous locations in ${\cal S}$ and $f(m+1)$
is the probability that these $m+1$ accesses are not to the cache block 
containing $d_{i,z}$. Figure~\ref{fig:ProcOutplaceMiss} shows the
other memory accesses between accesses $z$ and $z+1$ to $D_i$.

\begin{figure}
\epsfxsize 4in
\centerline{\epsfbox{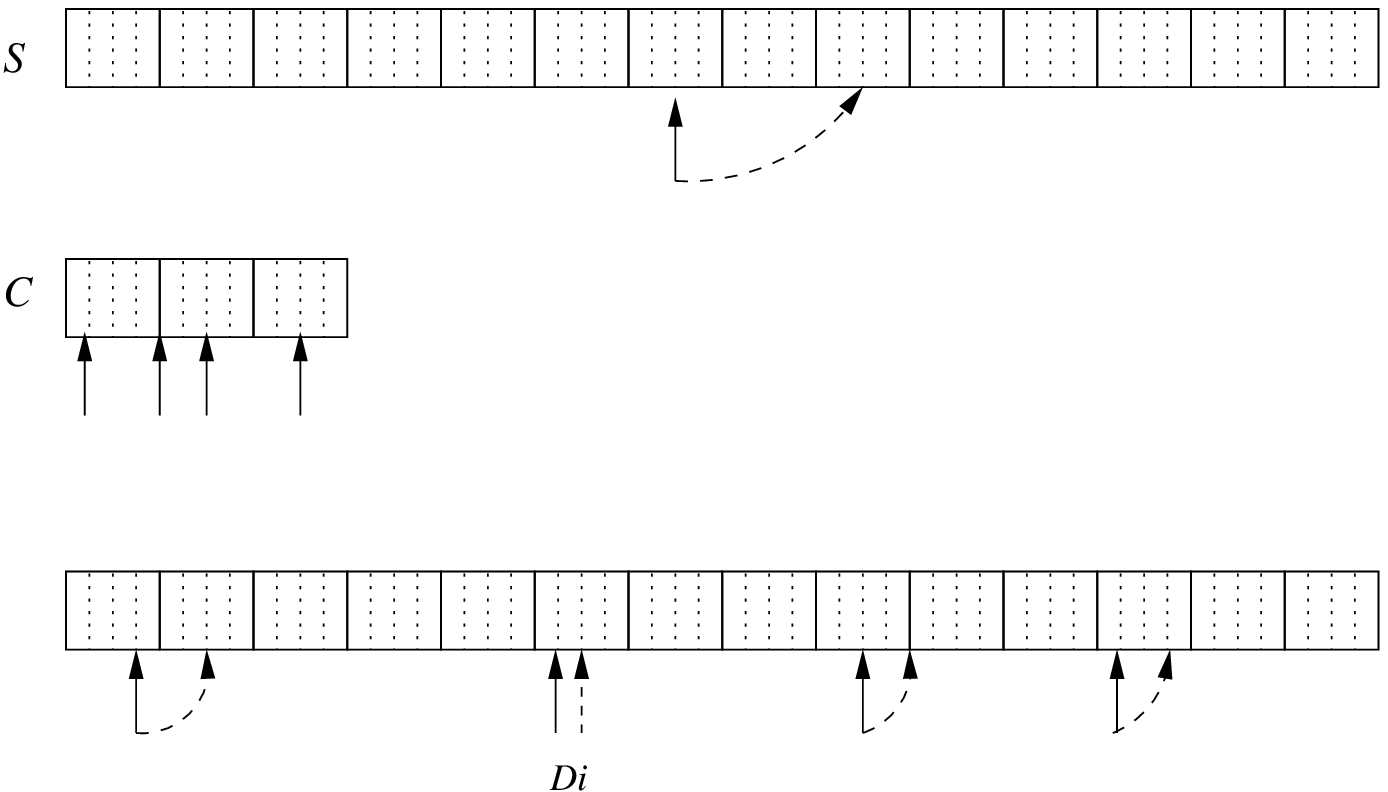}}
\caption{$m$ rounds of Process ``out-of-place". Between two accesses to
$D_i$, there are $m$ accesses to ``other" pointers, and $m+1$ accesses
to ${\cal C}$, and $m+1$ accesses to consecutive locations in ${\cal S}$. }
\label{fig:ProcOutplaceMiss}
\end{figure}

For a given configuration $\bar{m}$ of accesses,
as the probabilities $f(m_j)$, $g(\bar{m})$ and $f(m+1)$ are 
independent, we conclude that 
the probability that the cache block containing $d_{i,z}$ is
not accessed in these $m$ rounds is
$(1-g(\bar{m})) f(m) \prod_{j=1}^{k} f(m_j)$.  Averaging over all
configurations $\bar{m}$, we get that
\begin{eqnarray}
\Pr[X_i \mid \mu = m] = E_{\bar{\mu} \sim \varphi(m,\bar{a_i})}
[ (1- g(\bar{\mu})) f(m +1) \prod_{j=1}^{k}f(\mu_j)].  
\end{eqnarray}
Using which we get,
\begin{eqnarray}
\Pr[X_i] &=& \sum_{m=0}^\infty \Pr[\mu = m] \Pr[X_i | \mu = m] \nonumber\\
         &=& \sum_{m=0}^\infty  p_i(1-p_i)^m
E_{\bar{\mu} \sim \varphi(m,\bar{a_i})}\
\left[ (1- g(\bar{\mu})) f(m +1) \prod_{j=1}^{k}f(\mu_j) \right]. \nonumber\\
\label{eq:opp_pr_xi}
\end{eqnarray}

Arguing as for Eq.~\ref{eq:pd} we get that, other than for the first access, 
the probability $p_d$ of a cache miss for a pointer access is:
\begin{eqnarray}
p_d &=& \frac{1}{B} + \frac{B-1}{B} \sum_{i=1}^{k} p_i  ( 1-\Pr[X_i] ).
\label{eq:opp_pd}
\end{eqnarray}

Including the first access misses, the expected number of cache misses 
for pointer accesses is at most
\begin{equation}
\sum_{i=1}^{k} 1+ (np_i-1)
\left(\left(\frac{B-1}{B} (1-\Pr[X_i]) \right)+\frac{1}{B}\right) \le
np_d +k.
\label{eq:opp_ub_ptr}
\end{equation}

We now consider the probability of a cache miss for an access to 
a count array location.
Fix an $i \in \{1, \ldots, k/B\}$ and a $z \ge 1$ and consider
the probability of a miss between access $z$ and $z+1$ to count block
$c_i$. We define $\nu, \nu_j, m, 
\bar{\nu}, \bar{m}, P_i, \varphi(m,\bar{b_i})$, and $Y_i$ as in the proof 
of Theorem~\ref{thm:THEOREM1}.

Again, given that $k \le BC$ and assumption (c) 
mean that two blocks from ${\cal C}$ cannot conflict with each other.
So we need to determine the probability of a conflict 
given $m_j$ accesses to the pointer $D_j$, for all $j \in \{ 1, \ldots, k\}$, and 
$m$ accesses to contiguous locations in ${\cal S}$. 
Again $f(m_j)$ is the probability that none of the $m_j$ locations accessed by 
$D_j$ in $m$ rounds is mapped to the same cache block as $c_i$ and $f(m+1)$ is
the probability that the accesses to $m+1$ contiguous locations in ${\cal S}$ 
are not to the same cache block as $c_i$.

So we have:

\begin{eqnarray}
\Pr[Y_i] &=& \sum_{m=0}^\infty \Pr[\nu = m] \Pr[Y_i | \nu = m] \nonumber\\
         &=& \sum_{m=0}^\infty P_i(1-P_i)^m
 E_{\bar{\nu} \sim \varphi(m,\bar{b_i})} \
 \left[f(m) \prod_{j=1}^{k}f(\nu_j) \right].
\label{eq:opp_pr_yi}
 \end{eqnarray}
%

Arguing as for Eq.~\ref{eq:pc}, the probability $p_c$ of a cache miss for 
a count array access is:

\begin{equation}
p_c = \sum_{i=1}^{k/B} P_i  ( 1-\Pr[Y_i] ).
\label{eq:opp_pc}
\end{equation}
%

Including the first access misses, the expected number of cache misses 
for count array accesses is at most
\begin{equation}
\sum_{i=1}^{k/B} 1 + (nP_i-1) (1-\Pr[Y_i]) \le np_c +k/B.
\label{eq:opp_ub_count}
\end{equation}

We now calculate cache misses for accesses to the array ${\cal S}$. 
We consider the probability of a cache miss 
between accesses to ${\cal S}[s]$ and ${\cal S}[s+1]$. 
We know that there is exactly one access to a count 
block and one access to a pointer between two accesses to ${\cal S}$.
The probability that the pointer access is to the same cache block
as ${\cal S}[s]$ is $1/C$. The probability that a block from ${\cal C}$
maps to the same cache block as ${\cal S}[s]$  is $k/BC$.
Given that a block from ${\cal C}$
maps to the same cache block as ${\cal S}[s]$,
the probability that the access to the count array is to the same cache 
block as ${\cal S}[s]$ is $B/k$. So the probability that the pointer 
access is to the same cache block as ${\cal S}[s]$ is also $1/C$. 
So the probability that there are no memory accesses to the cache block that
${\cal S}[s]$ is mapped to before the access to ${\cal S}[s+1]$ is 
$$(1-1/C)^2.$$
We have a cache miss if ${\cal S}[s]$ is at a cache block boundary, otherwise
the probability of a cache miss is $1-(1-1/C)^2$.
So the probability $p_s$ of an cache miss for an access to ${\cal S}$ is
$$
p_s = \frac{1}{B} + \frac{B-1}{B}\left( 1 - \left(1-\frac{1}{C}\right )^2 \right).
$$

The first access to ${\cal S}$ is always a cache miss, so the expected number of
cache misses in accesses to ${\cal S}$ is:
$$
n p_s + 1.
$$

Plugging in the values from Eq.~\ref{eq:opp_pr_xi} into Eq.~\ref{eq:opp_ub_ptr} 
and from Eq.~\ref{eq:opp_pr_yi} into Eq.~\ref{eq:opp_ub_count} 
we get the upper bound on $X$, 
the expected number of cache misses in the processes.

The lower bound in Theorem~\ref{thm:OPP_THEOREM1} is obvious.

\end{proof}

\subsubsection{Upper bound}
We now prove a theorem on the upper bound to the expected
number of cache misses during $n$ rounds of Process ``out-of-place".

\begin{theorem}
\label{thm:OPP_THEOREM2}
The expected number of cache misses in $n$ rounds of Process ``out-of-place"
is at most $n(p_d + p_c + p_s) + k(1+1/B) +1$, where:
\begin{eqnarray*}
p_d &\le& \frac{1}{B} + \frac{2(B-1)k}{B^2C} + 
   \frac{B-1}{BC} \sum_{i=1}^{k} \sum_{j=1}^{k/B} \frac{p_iP_j}{p_i+P_j}  \\
   && + \frac{(B-1)^2}{B^2C} 
   \left( 1 + \sum_{i=1}^{k} \sum_{j=1}^{k} \frac{p_ip_j}{p_i+p_j}\right) \\
p_c &\le& \frac{2k}{B^2C} + \frac{B-1}{BC}\left ( 1 + \sum_{i=1}^{k/B} 
        \sum_{j=1}^{k}\frac{P_ip_j}{P_i+p_j} \right ), \\
p_s &=& \frac{1}{B} + \frac{B-1}{B} 
    \left(1 - \left( 1 - \frac{1}{C} \right)^2 \right). \nonumber
\end{eqnarray*}
\end{theorem}
\begin{proof}
As for the upper bound for in-place permutation, in this proof we derive 
lower bounds for $\Pr[X_i]$ and $\Pr[Y_i]$ and we will use these to 
derive the upper bounds on $p_d$ and $p_c$. We make extensive 
use of the results obtained during the proof of Theorem~\ref{thm:THEOREM1}.

Again, we consider a fixed $i$ and consider the event $X_i$ 
defined in the proof of Theorem~\ref{thm:OPP_THEOREM1}.
We now obtain a lower bound on $\Pr[X_i]$.\\
\\
{\bf Lower bound on $\Pr[X_i]$}\\
Letting $\Gamma(x)=1-f(x)$ and using Proposition\ref{prop:ProdSum}
we can rewrite Eq.~\ref{eq:pr_xi} as:
\begin{equation}
\Pr[X_i] \ge \sum_{m=0}^\infty \Pr[\mu = m]
                 \E_{\bar{\mu} \sim \varphi(m,\bar{a_i})}
                  \left [ 1  - g(\bar{\mu}) 
				             - \Gamma(m +1 )
                             - \sum_{j=1}^k \Gamma(\mu_j) \right].
\label{eq:opp_xiub_part2}
\end{equation}

We can use Eq.~\ref{eq:gmu_ub} as a simplification for
$$\sum_{m=0}^\infty \Pr[\mu = m] 
\E_{\bar{\mu} \sim \varphi(m,\bar{a_i})} [ g(\bar{\mu}) ],$$ 
and Eq.~\ref{eq:Gammamuj_ub} as an upper bound on
$$\sum_{m=0}^\infty \Pr[\mu = m]
\E_{\bar{\mu} \sim \varphi(m,\bar{a_i})}[ \sum_{j=1}^k \Gamma(\mu_j) ].$$
So we just have to evaluate 
$$\sum_{m=0}^\infty \Pr[\mu = m]
\E_{\bar{\mu} \sim \varphi(m,\bar{a_i})}[ \Gamma(m + 1) ].$$
Since we always have at least one access to ${\cal S}$, we have 
that
\begin{eqnarray}
\sum_{m=0}^\infty \Pr[\mu = m]
\E_{\bar{\mu} \sim \varphi(m,\bar{a_i})}[\Gamma(\mu +1)] &=&
\sum_{m=0}^{\infty} \Pr[\mu = m] \frac{m+B}{BC}  \nonumber \\
&-&\sum_{m=BC-B}^{m} \Pr[\mu = m] \left(  \frac{m+B}{BC} -1 \right) \nonumber\\
&\le& \frac{1}{BC} \left[ \sum_{m=0}^{\infty} \Pr[\mu=m] m + B\right] \nonumber\\
&=& \frac{1}{BC} \left(\frac{1}{p_i} - 1 + B \right),
\label{eq:Gammamu_ub}
\end{eqnarray}
where the last simplification used Proposition~\ref{prop:GeoSeriesc}(c).
Substituting Eq.~\ref{eq:gmu_ub}, Eq.~\ref{eq:Gammamuj_ub} and Eq.~\ref{eq:Gammamu_ub} 
in Eq.~\ref{eq:opp_xiub_part2} we obtain the following lower bound for $\Pr[X_i]$
\begin{eqnarray}
\Pr[X_i] &\ge& 1 - 
\frac{1}{C}\sum_{j=1}^{k/B}\frac{P_j}{p_i+P_j}  \nonumber\\
&& - \frac{1}{BC}
  \left ( \frac{2}{p_i} + (B-1) \left( 1 + \sum_{j=1}^{k}\frac{p_j}{p_i+p_j} \right) \right ).
\label{eq:opp_xiub_final}
\end{eqnarray}
\\
{\bf Upper bound on $p_d$}\\
Finally, substituting $\Pr[X_i]$ from Eq.~\ref{eq:opp_xiub_final} in Eq.~\ref{eq:pd} 
we get 
\begin{eqnarray}
p_d 
&\le& \frac{1}{B} +  
   \frac{B-1}{B} \sum_{i=1}^{k} p_i \nonumber\\
   &&\left ( 
   \frac{1}{C}\sum_{j=1}^{k/B} \frac{P_j}{p_i+P_j} +
   \frac{1}{BC}\left ( \frac{2}{p_i} + 
   (B-1)\left ( 1+\sum_{j=1}^{k} \frac{p_j}{p_i+p_j} \right)\right) \right) \nonumber\\
&=& \frac{1}{B} + \frac{2(B-1)k}{B^2C} + 
   \frac{B-1}{BC} \sum_{i=1}^{k} \sum_{j=1}^{k/B} \frac{p_iP_j}{p_i+P_j}  \nonumber\\
   && + \frac{(B-1)^2}{B^2C} \left( 1 + \sum_{i=1}^{k} \sum_{j=1}^{k} \frac{p_ip_j}{p_i+p_j}\right). \nonumber
\end{eqnarray}

We can evaluate $p_c$ using a very similar approach to that 
used in the proof of Theorem~\ref{thm:THEOREM2}. 
We again consider a fixed $i$ and consider the event $Y_i$ 
defined in the proof of Theorem~\ref{thm:OPP_THEOREM1}.
We now obtain a lower bound on $\Pr[Y_i]$.\\
\\
{\bf Lower bound on $\Pr[Y_i]$}\\
We can rewrite Eq.~\ref{eq:opp_pr_yi} as:
\begin{equation}
\Pr[Y_i] \ge \sum_{m=0}^\infty \Pr[\nu = m]
                 \E_{\bar{\nu} \sim \varphi(m,\bar{b_i})}
                  \left [ 1  - \Gamma(m + 1) \sum_{j=1}^k \Gamma(\nu_j) \right].
\label{eq:opp_yiub_part2}
\end{equation}
Eq.~\ref{eq:Gammanuj_ub} gives us 
$$
\sum_{m=0}^\infty \Pr[\nu=m] \sum_{j=1}^{k}
\E_{\bar{\nu} \sim \varphi(m,\bar{b_i})}[\Gamma(\nu_j)].
$$
%
%
Arguing as for Eq.~\ref{eq:Gammamu_ub}
\begin{eqnarray}
\sum_{m=0}^\infty \Pr[\nu = m]
\E_{\bar{\nu} \sim \varphi(m,\bar{b_i})}[\Gamma(m +1)]
%
%
&\le& \frac{1}{BC} \left(\frac{1}{P_i} - 1 + B \right).
\label{eq:Gammanu_ub}
\end{eqnarray}
%
%
Substituting Eq.~\ref{eq:Gammanuj_ub} and Eq.~\ref{eq:Gammanu_ub} 
in Eq.~\ref{eq:opp_yiub_part2} we obtain the following lower bound for $\Pr[X_i]$:
\begin{eqnarray}
\Pr[Y_i] &\ge& 1 - \frac{1}{BC}
  \left ( \frac{2}{P_i} + (B-1) \left( 1 + \sum_{j=1}^{k}\frac{p_j}{P_i+p_j} \right) \right ).
\label{eq:opp_yiub_final}
\end{eqnarray}
\\
{\bf Upper bound on $p_c$}\\
Finally, substituting $\Pr[Y_i]$ from Eq.~\ref{eq:opp_yiub_final} in Eq.~\ref{eq:pc} 
we get: 
\begin{eqnarray}
p_c &\le& \sum_{i=1}^{k/B} P_i \frac{1}{BC} \left( \frac{2}{P_i} + 
         (B-1)\left(1+ \sum_{i=1}^{k} \frac{p_j}{P_i+p_j} \right) \right) \nonumber\\
    &=& \frac{2k}{B^2C} + \frac{B-1}{BC}\left ( 1 + \sum_{i=1}^{k/B}
        \sum_{j=1}^{k}\frac{P_ip_j}{P_i+p_j} \right ). \nonumber
\end{eqnarray}

\end{proof}

\subsubsection{Lower bound}
It is quite obvious that the lower bound for in-place permutation, 
given in Theorem~\ref{thm:THEOREM3}, is a lower bound for out-of-place 
permutation.

\subsubsection{Upper and lower bounds for uniformly random data}
Using the upper bound Theorem just proven, we now derive a Corollary for 
an upper bound to the number of cache misses if the data is uniformly distributed. 

\begin{corollary}
\label{cor:OPP_UB_Cor1}
If $p_1 = \ldots = p_k = 1/k$ then the number of cache misses in $n$ rounds
of Process ``in-place" is at most:
$$
n\left( \frac{1}{B} + \frac{k(B+3)}{2BC} + \frac{k}{B^2C}+\frac{k}{BC}\right)
+ k\left(1+\frac{1}{B}\right).
$$
\end{corollary}
\begin{proof}
Since $P_i$ in $p_c$ and $P_j$ in $p_d$ are both $B/k$ in the equations in 
Theorem~\ref{thm:OPP_THEOREM2}, we get that
\begin{eqnarray}
p_d+p_c+p_s &\le& 
\frac{2}{B} + \frac{2(B-1)}{BC}\frac{k^2}{B}\frac{B/k}{B+1} 
+\frac{(B-1)^2}{B^2C}k^2\frac{1/k}{2} \nonumber\\
&&+ \frac{2k}{B^2C}+\frac{2(B-1)k}{B^2C} + \frac{B-1}{B}\left(1-\frac{(C-1)^2}{C^2}\right)\nonumber\\
&=& \frac{2}{B} + \frac{2(B-1)}{BC}\frac{k}{B+1} 
+\frac{(B-1)^2}{B^2C}\frac{k}{2} \nonumber\\
&&+ \frac{2k}{B^2C}+\frac{2(B-1)k}{B^2C} + \frac{B-1}{B}\frac{2C-1}{C^2}\nonumber\\
&\le& \frac{2}{B} 
+\frac{k}{C} \left[ \frac{4}{B}+\frac{B-1}{2B} \right]
+ \frac{2k}{B^2C} + \frac{2}{C}\nonumber\\
&=& \frac{2}{B} + \frac{k(B+7)}{2BC} 
+ \frac{2k}{B^2C} + \frac{2}{C}.\nonumber
\end{eqnarray}
\end{proof}

\begin{remark}
Corollaries~\ref{cor:UB_Cor1}~and~\ref{cor:OPP_UB_Cor1} shows that for 
uniformly distributed data, other than for small values of $B$,
the number of cache misses during in-place
and out-of-place permutations are quite close. As for an in-place
permutation, one pass of uniform distribution sorting 
using out-of-place permutations incurs $O(n/B)$ cache 
misses if and only if $k = O(C/B)$.

Using Corollary~\ref{cor:LB_Cor2} for the lower bound and 
Corollary~\ref{cor:OPP_UB_Cor1} above, we see that when $k \le C$
the lower bound is again within $3/2$ of the upper bound and is
much closer when $k \ll C$.
\end{remark}

\subsection{Cache Analysis of Multiple Sequences Access}
Accessing $k$ sequences is like Process ``in-place" in Section~\ref{proc:Inplace}
except that there is no interaction with a count array, so we delete step 2 
and assumption (c).
An analogue of Theorem~\ref{thm:THEOREM1} is easily obtained.
An easy modification to the proof of
Theorem~\ref{thm:THEOREM2} gives:
\begin{theorem}
The expected number of cache misses in $n$ rounds of sequence accesses
is at most:
$$k+n \left(\frac{1}{B} + \frac{k(B-1)}{B^2C} + 
\frac{(B-1)^2}{B^2C}\sum_{i=1}^{k}\sum_{j=1}^{k}\frac{p_ip_j}{p_i+p_j}\right).$$
\end{theorem}
\begin{corollary}
\label{cor:SEQ_COR1}
If $p_1 = \ldots = p_k = 1/k$ then the number of cache misses in $n$ rounds
of sequence accesses is at most:
$$n\left(\frac{1}{B} + \frac{k(B+3)}{2BC}\right) + k.$$
\end{corollary}
\begin{remark}
From Corollary~\ref{cor:SEQ_COR1},  $k = O(C/B)$ random sequences can be
accessed incurring an optimal $O(n/B)$ misses. 
This essentially agrees with the results obtained by Mehlhorn and 
Sanders~\cite{Mehlhorn2000} and Sen and Chatterjee~\cite{SC99}.
\end{remark}

%
%
\begin{remark}
Since its derivation ignored the effects of the count array, the lower
bound in Theorem~\ref{thm:THEOREM3} applies directly to sequence 
accesses.
Note that the lower bound we obtain for uniformly random data,
as stated in Corollary~\ref{cor:LB_Cor2}, is sharper than the
lower bound of $0.25(1-e^{-0.25k/C})$ obtained in~\cite{SC99}.
\end{remark}
%

\begin{remark}
Our upper and lower bounds are also closer than those in~\cite{Mehlhorn2000}.
The analysis in \cite{Mehlhorn2000} assumes that
accesses to the sequences are controlled by an adversary;
our analysis demonstrates, that with uniform randomised accesses to
the sequences, more sequences can be accessed optimally.
\end{remark}

\subsection{Correspondence between the processes and the permute phase}

We now show how the Processes ``in-place" and ``out-of-place" model the 
permute phase of a generic distribution sorting algorithm.

The correspondence between Process ``in-place" of Section~\ref{proc:Inplace}
and the pseudocode in Figure~\ref{fig:InplacePermuteCode} is as follows.  
Each iteration of the inner loop (steps 3.1-3.5) of the pseudocode
corresponds to a round of Process ``in-place".
The array {\tt COUNT\/} corresponds
to the locations $\cal C$, and the pointer $D_i$ points to 
\mbox{\tt DATA[{\it idx}]}.
The variables $x$ in the process and the pseudocode play
a similar role.  It can easily be verified that
in each iteration of the loop in the pseudocode,
the value of $x$ is any integer $1, \ldots, k$
with probability $p_1,\ldots,p_{k}$, independently of its
previous values, as in Step 1 of Process ``in-place".
A read at a location immediately followed by a write to the same 
location is counted as one access.
Thus, the read and increment of {\tt COUNT[{\it x}]} 
in Steps 3.2 and 3.3 of the pseudocode constitutes one access,
equivalent to Step 2.
Similarly the ``swap'' in Step 3.5
of the pseudocode corresponds to the memory access in Step 3
of the process.  The process does not model the
initial access in Step 1 of the pseudocode,
and nor does it model the task of looking for
new cycle leaders in Steps 4 and 5 of the pseudocode.

The correspondence between Process ``out-of-place" of Section~\ref{proc:Outplace}
and the pseudocode in Figure~\ref{fig:OutofplacePermuteCode} is as follows.  
The array {\tt COUNT\/} corresponds
to the locations $\cal C$, the array {\tt DATA\/} corresponds
to the locations $\cal S$, $i$ in the pseudocode corresponds to 
$s$ in the process, and the pointer $D_i$ points to
\mbox{\tt DEST[{\it idx}]}. 
The increment of $i$ in the pseudocode is equivalent to the increment 
of $s$ in the process, and the accesses to {\tt DATA[{\it i}]} 
and ${\cal S}[s]$ are equivalent.
As above, the variables $x$ in the process and the pseudocode in 
the pseudocode play a similar role.
Again, the read and increment of {\tt COUNT[{\it x}]} 
of the pseudocode constitutes one access, and is equivalent to Step 
3 of the process. 
The access to \mbox{\tt DEST[{\it idx}]} of the pseudocode corresponds 
to the memory access in Step 4 of the process.  

Assumption (b) of the processes is clearly satisfied and 
assumption (c) can normally be made to hold.
Assumption (d) and $k\le CB$ or $k\le C$ may not hold in practice, 
in~\cite{RR99} we give an approximate analysis which deals with this.
Assumption (a) of the processes,
that the starting locations of the pointers
$D_i$ are uniformly and independently distributed, is patently false,
we discuss this in more detail in~\cite{RR99}.
We may force it to hold by adding
random offsets to the starting location of each pointer,
at the cost of needing more memory and adding a compaction phase 
after the permute, this has also been suggested by Mehlhorn and
Sanders~\cite{Mehlhorn2000}.
This only works if the permute is not in-place, and if
$k$ is sufficiently small (e.g. $k \le n/(CB)$). 
In~\cite{RR99} we study assumption (a)
empirically in the context of uniform distribution sorting.
Another weakness is that our processes are continuous,
so the sequence lengths are not specified, whereas in distribution
sorting we sort $n$ keys and each sequence is of a finite length.

\section{MSB radix sort}

We now consider the problem of sorting $n$ independent and
uniformly-distributed floating-point numbers in the range $[0,1)$
using the integer sorting algorithm MSB radix sort.  
As noted earlier, it suffices to sort lexicographically the bit-strings 
which represent the floats, by viewing them as integers. 
One pass of MSB radix sort using radix size $r$ groups the keys 
according to their most significant $r$ bits in $O(2^r + n)$ time.
For random integers, a reasonable choice for minimising instruction counts 
is $r = \lceil \log n - 3 \rceil$ bits, or classifying into
about $n/8$ classes.  Since each class has about 8 keys on average, they
can be sorted using insertion sort.  Using this approach for
this problem gives terrible performance even at small values of $n$
(see Table~\ref{tab:QSvsNaiveH}).
As we now show, the problem lies with the 
distribution of the integers on which MSB radix sort is applied.

\noindent
\subsection{Radix sorting floating-point numbers}
A floating-point number is represented as a triple of non-negative
integers $\langle i, j, l \rangle$.  Here $i$ is called
the {\it sign bit} and is a 0-1 value (0 indicating non-negative
numbers, 1 indicating negative numbers), $j$ is called the 
{\it exponent\/} and is represented using $e$ bits and 
$l$ is called the {\it mantissa} and represented using $m$ bits.
Let $j^{*} = j - 2^{e-1} + 1$ denote the {\it unbiased\/}
exponent of $\langle i, j, l \rangle$.
Roughly following the IEEE 754 standard,
let the triple $\langle 0, 0, 0\rangle$ represent the 
number 0, and let $\langle i, j, l\rangle$, where $j > 0$,
represent the number $\pm 2^{j^{*}}(1 + l 2^{-m})$, depending
on whether $x = 0$ or $1$;  no other triple is a 
floating-point number. Internally each member of 
the triple is stored in consecutive
fields of a word.  The IEEE 754 standard specifies 
$e=8$ and $m=23$ for 32-bit floats and
$e=11$ and $m=52$ for 64-bit floats~\cite{Hennessy}.


We model the generation of a random float in the range $[0, 1)$
as follows: generate an (infinite-precision)
random real number, and round it down to the next smaller
float.  On average, half the numbers generated will lie
in the range $[0.5, 1)$ and will have an unbiased exponent of $-1$.
In general, for all non-zero numbers,
the unbiased exponent has value $i$ with probability $2^i$,
for $i =  -1, -2, \ldots, -2^{e-1}+2$, whereas the mantissa 
is a random $m$-bit integer.  The value $0$ 
has probability $2^{-2^{e-1}+2}$.   Clearly, the distribution is not uniform,
and it is easy to see that the average size of the largest
class after the first pass of MSB radix sort with radix $r$ is
$n \left ( 1 - \frac{1}{2^{2^{e-r+1}}} \right )$  if $r < e+1$,
and $n/(2^{r - e})$ if $r \ge e+1$.

This shows, e.g., that the largest sub-problems in the
examples of Table~\ref{tab:QSvsNaiveH} would be of size 
$n/2^{\lceil \log n - 3 \rceil - 11} \approx 2^{14}$,
so using insertion sort after one pass is 
inefficient in this case\footnote{In fact, the total number of keys in all sub-problems of this size would be $n/2$ on average.}.   
To get down to problems of size $8$ in one pass requires a radix
of about $\log n + 8$,  which is impractical.
Also, MSB radix sort applied to random integers has $O(n)$ expected
running time independently of the word size, but this
is not true for floats.  A first pass 
with $r \ll e$ barely reduces the largest problem size, and
the same holds for subsequent passes until bits from
the mantissa are reached.  As the radix in any 
pass is limited to $\log n + O(1)$ bits, we may need
$\Omega(e/\log n)$ passes, introducing
a dependence on the word size.

\subsection{Using Quicksort}
To get around the problem of having several passes before we reduce 
the largest class, we partition the input keys around a value
$1/n \le \theta \le 1/(\log n)$, and sort the keys smaller
than $\theta$ in $O(n)$ 
expected time using Quicksort.  We then apply MSB
radix sort to the remaining keys.  Let 
$e' =  \min\{\lceil \log \log (1/\theta) \rceil, e\}$
denote the {\it effective exponent}, since the remaining keys have exponents
which vary only in the lower order $e'$ bits.
This means that keys can be grouped according to
a radix $r = e + 1 + m'$ with $m' \ge 0$ in
$O(n + 2^{e' + m'})$ time and $O(2^{e' + m'})$ space. 
Since $e' = O(\log\log n)$, we can take up to $\log n - O(\log \log n)$ bits from
the mantissa as part of the first radix;  
as all sub-problems now only deal with bits from the mantissa
they can be solved in linear expected time, giving a
linear running time overall.  

\subsection{Cache analysis}
We now use our analysis to
calculate an upper bound for the cache misses 
in the permute phase of the first pass of MSB radix sort 
using a radix $r = e + 1 + m'$, for some $m' \ge 0$, assuming also
that all keys are in the range $[\theta, 1)$, for some $\theta \ge 1/n$.
There are $2^{e' + m'}$ pointers in all, which
can be divided into $g = 2^{e'}$ {\it groups\/} of
$K = 2^{m'}$ pointers each.  Group $i$ corresponds to
keys with unbiased exponent $-i$, for $i = 1,\ldots,g$. 
All pointers in group $i$ have an access probability of $1/(K2^i)$.
Using Theorem~\ref{thm:THEOREM1} and a slight extension of
the methods of Theorem~\ref{thm:THEOREM2} we are able to
prove Theorem~\ref{thm:MSBmisses} below, which states
that the number of misses is essentially independent of
$g$:

\begin{theorem}
\label{thm:MSBmisses}
Provided $gK\le CB$ and $K \le C$ the number of misses in the
first pass of the permute phase of MSB radix sort is at most:
$$
n \left(\frac{1}{B} + \frac{2K}{BC} 
  \left( 2.3 B + 2\log B + \log C - \log K + 0.7 \right) \right) +
   gK(1+1/B).
$$
\end{theorem}

\begin{proof}
Using Eq.~\ref{eq:xiub_final} we can calculate
an upper bound on the probability of event $X_{(i-1)K+1}$ as:

\begin{eqnarray}
\lefteqn{\Pr[X_{(i-1)K+1}] } \nonumber \\
& \ge & 1 -\frac{K2^i}{BC}
-\frac{1}{C}   \sum_{j=1}^{g}\sum_{l=1}^{K/B}\frac{B2^{-j}}{B2^{-j}+2^{-i}}
-\frac{B-1}{BC}\sum_{j=1}^{g}\sum_{l=1}^{K}\frac{2^{-j}}{2^{-j}+2^{-i}} \nonumber \\
&=& 1 - \frac{K}{BC} \left (
     \sum_{j=1}^{g}\frac{B2^{-j}}{B2^{-j}+2^{-i}} + 2^i +
(B-1)\sum_{j=1}^{g}\frac{2^{-j}}{2^{-j}+2^{-i}}  \right ) \nonumber \\
&\ge& 1 - \frac{K}{BC} \left ( 
 \log B + i + 2^i + (B-1)\sum_{j=1}^{g}\frac{2^{-j}}{2^{-j}+2^{-i}} \right ).
\label{eq:MSB_Xiub}
\end{eqnarray}

If $K2^i/(BC) \ge 1$ then $\Pr[X_{(i-1)K+1}]$ would be negative,
so we place a bound on this term such that $K2^i < BC$. The maximum value
of $i$ such that $K2^i/(BC) < 1$ is $\log BC - \log K - 1$.

Since the probabilities of access to pointer $D_{(i-1)K+1}, \ldots, D_{iK}$ 
are all $1/(K2^i)$ we can calculate an upper bound on $p_d$ using
Eq.~\ref{eq:pd}~and~\ref{eq:MSB_Xiub} as:
\begin{eqnarray*}
p_d &\le& \frac{K^2}{BC} \left(\sum_{i=1}^{g} p_i \left ( 
\log B + i + (B-1)\sum_{j=1}^{g}\frac{2^{-j}}{2^{-j}+2^{-i}} \right )
+ \sum_{i=1}^{\log BC - \log K - 1} p_i 2^i \right) \\
&&+ \sum_{i=\log BC - \log K}^{g} K p_i \\
&=& \sum_{i=1}^{g} \frac{1}{2^i} \frac{K}{BC} \left ( 
\log B + i + (B-1)\sum_{j=1}^{g}\frac{2^{-j}}{2^{-j}+2^{-i}} \right ) \\
&&+ \sum_{i=1}^{\log BC - \log K -1} \frac{1}{2^i} \frac{K}{BC} 2^i
+ \sum_{i=\log BC-\log K}^{g} \frac{1}{2^i} \\
&\le& \frac{K}{BC} \left(\
\sum_{i=1}^{g} \frac{\log B}{2^i} +
\sum_{i=1}^{g} \frac{i}{2^i} +
(B-1) \sum_{i=1}^{g} \frac{1}{2^i} \sum_{j=1}^{g}\frac{2^{-j}}{2^{-j}+2^{-i}} 
\right ) \\
&&+ \frac{K}{BC} \left(\log BC-\log K - 1 \right) + \frac{2K}{CB} \\
&\le& \frac{K}{BC} \left( 2\log B + 3 + \log C - \log K + 2.3(B-1) \right ).
\label{eq:MSB_pdub}
\end{eqnarray*}

\end{proof}


\smallskip
\noindent
\subsection{Tuning MSB radix sort}
We now optimise parameter choices in our algorithms.
The smaller the value of $\theta$, the fewer keys
are sorted by Quicksort, but reducing $\theta$ may
may increase $e'$.  A larger value of $e'$ does not
mean more misses, by Theorem~\ref{thm:MSBmisses}, but it
does mean a larger count array.
We choose $\theta = 1/(\log n)^2$ as a compromise,
ensuring that Quicksort uses $o(n)$ time.
Using the above analysis 
we are also able to determine
an optimal number of classes to use in each sorting sub-problem.
We use two criteria of optimality.  In the first, we require that
each pass incur no more than $(2 + \varepsilon)n/B$ misses for
some constant $\varepsilon > 0$, thus seeking essentially to
minimise cache misses  ($2n/B$ misses is the bare minimum for 
the count and permute phases). In the second, we trade-off
reductions in cache misses against extra computation.  The
latter yields better practical results, and results shown below
are for this approach.

\subsection{Experimental results}
Table~\ref{tab:all_float} 
compares tuned MSB radix sort with memory-tuned Quicksort\cite{LL97}
and MPFlashsort~\cite{RR99}, a memory-tuned version of a distribution 
sorting algorithm which assumes that the keys are independently drawn
from a uniformly random distribution. 
The algorithms were coded in C and
compiled using {\tt gcc 2.8.1}.  The experiments were our
Sun UltraSparc-II with $2 \times 300$ MHz processors and 1GB main
memory, and a 
16KB L1 data cache, 512KB L2 direct-mapped cache.  Observe that 
MSB radix sort easily outperforms the other algorithms for the 
range of values considered.
\begin{table}
{
\small
\begin{center}
\begin{tabular}{|l|r|r|r|r|r|r|}
\hline
\multicolumn{1}{|r|}{$n=$}
           & $1 \times 10^6$ & $2 \times 10^6$ & $4 \times 10^6
$ & $8 \times 10^6$ & $16 \times 10^6$ & $32 \times 10^6$\\
\hline
MTQuick  & 0.7400  & 1.5890  &  3.3690 & 7.2430  & 15.298 & 32.092 \\
\hline
Naive1 & 7.0620  & 14.192  & 28.436  & 57.082  & 115.16 & 233.16 \\
\hline
\end{tabular}
\end{center}
}
\caption{Memory-tuned Quicksort and Naive1 MSBRadix. 
Running times in seconds of memory-tuned Quicksort and Naive1 MSBRadix sort
(single pass MSBRadix sort without partitioning,
$r=\lceil \log n - 3 \rceil$)
floating point keys.}
\label{tab:QSvsNaiveH}
\end{table}
\begin{table}
{
\begin{center}
\begin{tabular}{|l|r|r|r|r|r|r|r|}
\hline
\multicolumn{1}{|r|}{$n=$}  & $1 \times 10^6$ & $2 \times 10^6$ & $4 \times 10^6$ & $8 \times 10^6$ & $16 \times 10^6$ & $32 \times 10^6$ & $64 \times 10^6$ \\
\hline
MPFlash & 0.6780 & 1.3780 & 2.2756 & 6.1700 & 13.308 & 27.738 & 56.796 \\
\hline
MTQuick   & 0.7400 & 1.5890 & 3.3690 & 7.2430 & 15.298 & 32.092 & 67.861 \\
\hline
MSBRadix& 0.3865 & 0.8470 & 1.9820 & 5.0300 & 9.4800 & 19.436 & 40.663 \\
\hline
\end{tabular}
\end{center}
}
\caption{MPFlashsort, memory-tuned Quicksort and MSBRadix.
Running times in seconds of MPFlashsort, 
memory-tuned Quicksort and MSBRadix sort on a Sun UltraSparc-II 
using single precision floating point keys.}
\label{tab:all_float}
\end{table}


\section{Conclusions}
We have analysed the average-case cache performance of 
the permute phase of distribution sorting when the keys are 
independently but not uniformly distributed. We have presented 
equations for the number of misses during in-place and out-of-place
permutations and have given closed-form upper and lower bounds
on these. We have shown that the upper and lower bounds are quite 
close when $k\le C$ and the data is known to be independently and
uniformly distributed.
We have shown how this analysis can easily be extended 
to obtain the number of cache misses during accesses to
multiple sequences.

We have shown that if the integer sorting algorithm MSB radix sort is used
to sort uniformly and randomly distributed floating point numbers then
a non-uniform distribution of keys to classes is induced. 
We have shown that a naive implementation of this algorithm would have
very poor performance due to this non-uniform distribution.
We have shown that by partitioning the keys, to remove keys which
are expected to go into small classes, and by using our analysis, 
the algorithm can be tuned for good cache performance. 
Due to fast integer operations and good cache utilisation the tuned 
algorithm outperforms MPFlashsort, a cache-tuned distribution sorting
algorithm, and memory-tuned Quicksort.

\end{document}